\documentclass[journal]{IEEEtran}
\usepackage{graphicx}
\usepackage{cite}
\usepackage{amsmath,amssymb,amsfonts,amssymb}
\usepackage{graphicx}
\usepackage{epstopdf}
\usepackage{textcomp}
\usepackage{xcolor}
\usepackage{mathrsfs}
\usepackage{caption}
\usepackage{array}
\allowdisplaybreaks  
\usepackage{graphicx}
\usepackage{verbatim}
\usepackage{enumitem}
\usepackage{booktabs}
\usepackage{tabularx}
\usepackage{bm}
\usepackage{color}
\usepackage{hyperref}
\hypersetup{hidelinks}

\usepackage{cite}
\usepackage{amsmath}
\usepackage{psfrag}
\usepackage{amsfonts}
\usepackage{graphicx}
\usepackage{multirow}
\usepackage{float}
\usepackage{color} 
\usepackage{stfloats}
\usepackage{balance}
\usepackage{bm}
\usepackage{cuted}
\usepackage{dutchcal}
\usepackage[none]{hyphenat}
\usepackage[lined,ruled,commentsnumbered]{algorithm2e}
\usepackage{diagbox}
\usepackage{array}
\usepackage{booktabs}
\usepackage{url}
\usepackage{subfigure}
\usepackage{caption}
\allowdisplaybreaks[4]

\newtheorem{lemma}{Lemma}

\newtheorem{example}{Example}
\newtheorem{remark}{Remark}
\newtheorem{assumption}{Assumption}

\newcommand{\F}{\mathbf{F}}

\newcommand{\I}{\mathbf{I}}
\newcommand{\n}{\mathbf{n}}
\renewcommand{\v}{\mathbf{v}}
\newcommand{\p}{\mathbf{p}}

\newcommand{\w}{\mathbf{w}}

\newcommand{\B}{\mathbf{B}}
\newcommand{\z}{\mathbf{z}}
\renewcommand{\r}{\mathbf{r}}

\newcommand{\diag}{\textup{diag}}
\newcommand{\G}{\mathbf{G}}

\long\def\symbolfootnote[#1]#2{\begingroup\def\thefootnote{\fnsymbol{footnote}}
\footnote[#1]{#2}\endgroup}

\begin{document}

\title{\LARGE Deep Reinforcement Learning Based Block Coordinate Descent for Downlink Weighted Sum-rate Maximization on AI-Native Wireless Networks
}


  \author{ Siya~Chen, Chee~Wei~Tan and~H. Vincent Poor 
\thanks{S. Chen is with School of Computer Science and Technology, Dongguan University of Technology, Dongguan, China and was with the Department of Computer Science, City University of Hong Kong, Hong Kong (e-mail: chensiya@dgut.edu.cn).}
 \thanks{C. W. Tan is with Nanyang Technological University, Singapore, Nanyang Ave., Singapore (e-mail: cheewei.tan@ntu.edu.sg).}
\thanks{H. V. Poor is with the Department of Electrical and Computer Engineering, Princeton University, Princeton, NJ 08544 USA (e-mail: poor@princeton.edu).}
\thanks{This work is supported in part by an Innovation Grant from Princeton NextG, and the Singapore Ministry of Education Academic Research Fund MOE-T2EP20224-0009.}
}

\maketitle

\begin{abstract}
This paper introduces a deep reinforcement learning-based block coordinate descent (DRL-based BCD) algorithm to address the nonconvex weighted sum-rate maximization (WSRM) problem with a total power constraint. Firstly, we present an efficient block coordinate descent (BCD) method to solve the problem. While this method may not always achieve globally optimal solutions, it provides a pathway for integrating machine learning and domain-specific techniques with theoretical analysis of the underlying convexity of the subproblems. We then integrate deep reinforcement learning (DRL) techniques into the BCD method and propose the DRL-based BCD algorithm. This approach combines the data-driven learning capability of machine learning techniques with the navigational and decision-making characteristics of the optimization-theoretic-based BCD method. This combination significantly improves the algorithm's performance by reducing its sensitivity to initial points and mitigating the risk of entrapment in local optima.
The primary advantages of the proposed DRL-based BCD algorithm lie in its ability to adhere to the constraints of the WSRM problem and significantly enhance accuracy, potentially achieving the exact optimal solution. Moreover, unlike many pure machine-learning approaches, the DRL-based BCD algorithm capitalizes on the underlying theoretical analysis of the WSRM problem's structure. This enables it to be easily trained and computationally efficient while maintaining a level of interpretability. 
Moreover, the DRL-based BCD framework demonstrates strong extensibility and can effectively be applied to other scenarios, such as joint beamforming for sum rate maximization, as demonstrated in this paper.
Through numerical experiments, the DRL-based BCD algorithm demonstrates substantial advantages in effectiveness, efficiency, robustness, and interpretability for maximizing sum rates, which also provides valuable potential for designing resource-constrained AI-native wireless optimization strategies in next-generation wireless networks.
\end{abstract}

\IEEEoverridecommandlockouts


\IEEEpeerreviewmaketitle

\section{Introduction}\label{sec:Introduction}

The growing complexity and demands of modern communication systems require the rapid advancement of wireless communication technologies to facilitate the development of more intelligent and adaptive wireless network architectures. AI-native wireless networks \cite{o2017introduction,hoydis2021toward,song2022benchmarking,lin2023fundamentals,chen2025openranet}, which integrate artificial intelligence (AI) across all layers of the network stack, present a promising solution for enhancing communication efficiency in large-scale systems. A promising area of exploration within AI-native wireless networks is the use of deep learning techniques for real-time decision-making and adaptive resource allocation that can respond to fluctuating user demands \cite{o2017introduction,hoydis2021toward,song2022benchmarking,lin2023fundamentals,chen2025openranet,hauffen2025deep}. One of the significant challenges that deep learning shows promise in addressing is maximizing sum rates, or total data throughput for multiple users within the wireless system \cite{sun2018learning,cui2019spatial,lee2020graph,mismar2019deep,waraiet2023robust,nasir2019multi,khan2020centralized,rahmani2022deep,vishnoi2023deep}. Maximizing the sum rate has long been a focal point in wireless network optimization research \cite{ chiang2007power, shi2011iteratively, zhang2018performance, shen2018fractional, ZhengTanJSAC213, TanTSP11,sun2018learning,cui2019spatial,lee2020graph,mismar2019deep,waraiet2023robust,nasir2019multi,khan2020centralized,rahmani2022deep,vishnoi2023deep,chen2023neural}. By employing deep learning algorithms, these systems can leverage large volumes of data to analyze and predict user behavior and channel conditions, enabling more efficient resource allocation that enhances the overall quality of service and ultimately improves user satisfaction.

This paper focuses on leveraging deep learning techniques to address the weighted sum rate maximization (WSRM) problem for AI-native wireless networks. Specifically, we examine the WSRM problem in downlink systems, where the base station transmits data to mobile devices equipped with multiple transmit and receive antennas under conditions where transmission power is a limited resource. This problem is important in practical applications because it enables the efficient use of limited transmission power resources and helps ensure compliance with wireless communication regulatory standards \cite{TanTSP11,Cai11,kim2020optimized}. Furthermore, it allows for an increase in the overall capacity of the system and facilitates effective management of interference levels. However, addressing this problem presents significant challenges due to its non-convex nature. Additionally, as the number of users increases, the task of efficiently allocating resources among them while maintaining computational efficiency becomes more complex.

Over the years, various algorithms for solving the sum rate maximization (SRM) problem have been developed, as presented in \cite{ chiang2007power, shi2011iteratively, zhang2018performance, shen2018fractional, ZhengTanJSAC213, TanTSP11}. These algorithms aim to strike a balance between computational complexity and achieving maximum overall system throughput. Recent advances in deep learning within AI-native wireless networks enable efficient end-to-end optimization, making these techniques an appealing strategy for enhancing wireless performance. Numerous studies have explored the SRM problem using deep learning techniques. A common approach involves using artificial neural networks to create models that map network states-such as channel conditions and user demands-to optimal resource allocation strategies. Once trained, these models can effectively allocate resources in response to new network states. Techniques studied in the context of sum rate maximization include deep neural networks (DNNs) \cite{sun2018learning}, convolutional neural networks (CNNs) \cite{cui2019spatial}, and graph neural networks (GNNs) \cite{lee2020graph}. This paper aims to contribute to this growing body of research by examining the effectiveness of these techniques in optimizing the sum rate in wireless communication systems.

Deep reinforcement learning (DRL) is another learning-based model that has been applied to sum rate maximization \cite{mismar2019deep,waraiet2023robust,nasir2019multi,khan2020centralized,rahmani2022deep,vishnoi2023deep}. The goal of DRL is to update the learning parameters iteratively to discover an optimal policy decision that maximizes long-term performance.  In the context of the SRM problem, the agent, usually representing the SRM problem or users, takes actions (e.g., power allocation, user scheduling) according to the policy, and receives rewards (e.g., sum-rate) based on the model performance. The agent then updates its policy based on the received rewards to improve its future actions \cite{mismar2019deep,nasir2019multi,khan2020centralized,waraiet2023robust,rahmani2022deep,vishnoi2023deep}. 
Various DRL-based power control algorithms have been applied to the SRM problem. 
Previous studies, such as \cite{mismar2019deep} and \cite{rahmani2022deep}, used single-agent DRL models, while \cite{nasir2019multi, khan2020centralized, vishnoi2023deep} proposed distributed dynamic algorithms employing multiple-agent DRL for power allocation.
\cite{nasir2019multi} and \cite{mismar2019deep} applied deep Q-learning, while \cite{rahmani2022deep,vishnoi2023deep} utilized deep deterministic policy gradient (DDPG), \cite{khan2020centralized} adopted trust region policy optimization (TRPO), and \cite{rahmani2022deep,waraiet2023robust} employed twin delayed deep deterministic policy gradient (TD3) to solve SRM problems.

 However, the general idea in most of the above DRL-based works is to create a sequence of states by using the wireless system parameters and then apply different  DRL techniques to learn power allocation actions. Due to the complexity and non-convexity of SRM problems, DRL models alone struggle to capture essential features without leveraging analytical properties, limiting their accuracy.
Additionally, pure machine learning methods require extensive parameter tuning and large datasets for large-scale networks, leading to longer training times and higher resource costs, which further challenge DRL model design and performance. 
Thus, while machine learning methods show promise, they are unlikely to fully replace traditional optimization-based approaches.
Integrating modern machine learning with traditional optimization offers a more promising approach to addressing the complex wireless optimization \cite{hoydis2021toward,o2017introduction,song2022benchmarking}. 
Therefore, more research is needed to explore the integration of modern machine learning with optimization methods to effectively address the SRM problem.

\begin{figure}
\centerline{\includegraphics[scale=0.38]{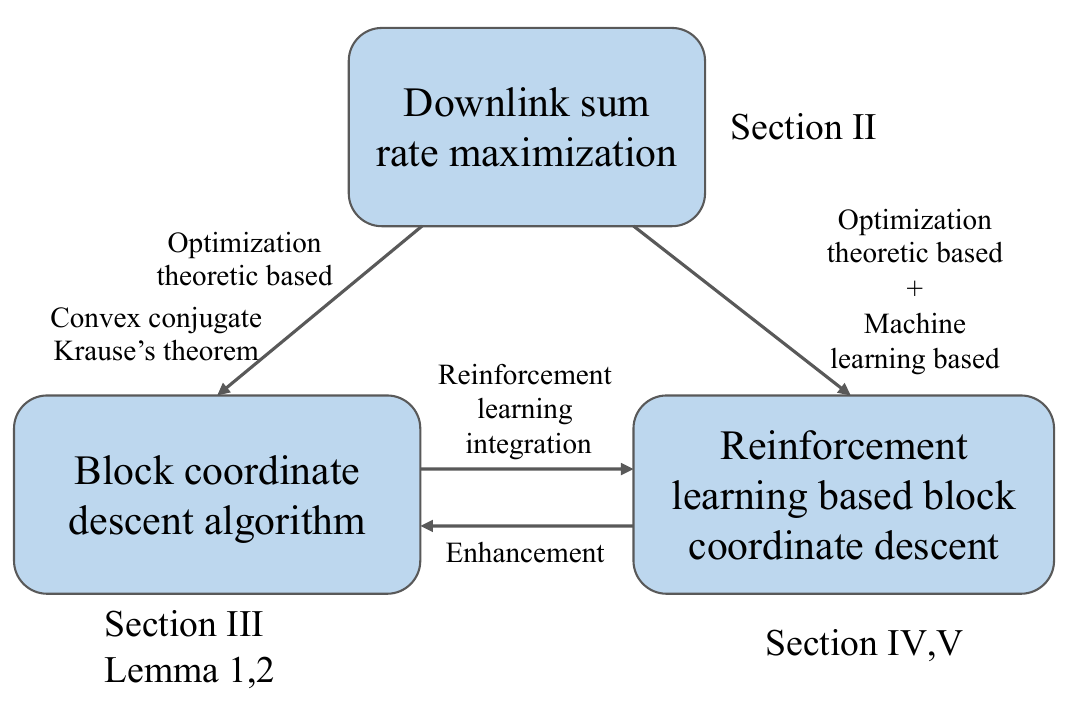}}
\captionsetup{font={footnotesize, stretch=1}, justification=raggedright}
\caption{ An overview of the relationship between the WSRM problem and the reinforcement learning based block coordinate descent algorithm.}
\label{fig:workflow}
\end{figure}

In this paper, we aim to integrate DRL models with optimization-theoretic techniques to solve the WSRM problem. One of the major advantages of this integration is the significant enhancement of algorithm accuracy, as it enables our algorithms not only to learn from data but also to exploit the convexity of the problem's underlying structure.
We begin by introducing an efficient BCD method for the WSRM problem. The BCD method was comprehensively analyzed by P. Tseng \cite{tseng2001convergence} and has also been employed in wireless power allocation \cite{pan2017joint,li2021sum}. While the proposed BCD method may not always achieve global optima, it effectively captures the intrinsic convexity of the original problems through theoretical analysis of convex subproblems, providing a basis for integrating deep learning with optimization techniques. Thus, we propose the DRL-based BCD algorithm, which combines data-driven DRL techniques with the BCD method's theoretical insights to tackle the WSRM problem.
Fig. \ref{fig:workflow} gives an overview of the problems and the DRL-based BCD algorithm.  
The motivation behind this approach is as follows:
1) Incorporating DRL techniques allows our algorithm to leverage data knowledge to uncover superior solutions;
2) BCD's navigational and decision-making characteristics align well with the  DRL algorithms, where the agent interacts with the environment to maximize rewards;
3) While BCD iterations are deterministic, poorly selected initial points may lead to local optima. Introducing DRL with exploration noise allows BCD to escape these local optima and reduces sensitivity to initial values;
4) Integrating BCD's theoretical analysis into DRL preserves the WSRM problem's analysis structures, which simplifies model design and training while enhancing interpretability compared to purely data-driven approaches;
and 5) Our proposed BCD method is highly efficient, and this efficiency carries over to the DRL-based BCD algorithm.
In summary, the significant contributions of our paper are as follows:
\begin{enumerate}
\item 
Firstly, we address the WSRM problem using the BCD method, which ensures convergence to a local optimum with high computational efficiency, though it may not always achieve global optimality.
\item 
Secondly, we propose a DRL-based BCD algorithm that combines the learning capabilities of DRL with the theoretical analysis of BCD. Its main contributions include guaranteeing the conditions for sum-rate maximization and significantly improving accuracy, potentially achieving the exact optimal solution. This is due to the integration of efficient iterations of optimal solutions for convex subproblems into the DRL state updating process, as well as the policy network's output for constructing those subproblems. Additionally, the DRL-based BCD algorithm is easier to train and more computationally efficient than pure DRL or optimization-based methods, while still maintaining interpretability.

\item Furthermore, we provide complexity analysis and convergence analysis for the DRL-based BCD, which offers theoretical assurances for its practical implementation. 
\textcolor{black}{
\item 
We also conduct additional analysis showing that the DRL-based BCD framework can be effectively extended for joint beamforming to maximize the sum rate, demonstrating the framework's strong scalability potential.}


\item Finally, we present numerical results demonstrating the performance of our DRL-based BCD algorithm, incorporating DDPG and TD3 techniques, which form the DDPG-based and TD3-based BCD algorithms, respectively. Additionally, we demonstrate that our algorithm outperforms baseline alternatives.
\end{enumerate}

The rest of the paper is organized as follows: In Section~\ref{sec:model}, we introduce the WSRM problem. In Section~\ref{sec:BCD}, we present the BCD method to solve the SRM problem. Section~\ref{sec:RL-based BCD} proposes the DRL-based BCD algorithm, which integrates DRL and BCD to solve the WSRM problem. Section~\ref{sec:extend} extends the application of DRL-based BCD for solving downlink WSRM beamforming. In Section~\ref{sec:expm}, we present the numerical simulations. Finally, Section~\ref{sec:conclusion} concludes the paper.

\section{System Model and Problem Formulation}\label{sec:model}
\textcolor{black}{
We consider a MIMO downlink system where $L$ independent data streams are transmitted over a shared frequency band. The transmitter is equipped with $N$ antennas, while the receiver for the $l$-th stream, denoted as $r_l$, has $N_l$ antennas. We assume that linear transmit and receive beamforming techniques are employed. The transmitted signal for the $l$-th stream is represented by $\mathbf{x}_l \in \mathbb{C}^{N \times 1}$, while the received signal at $r_l$, denoted as $\mathbf{y}_l \in \mathbb{C}^{N_l \times 1}$, is given by \cite{tan2015wireless}:
$$
\mathbf{y}_l = \mathbf{H}_l \mathbf{x} + \mathbf{z}_l, \quad l = 1, \ldots, L,
$$
where $\mathbf{H}_l \in \mathbb{C}^{N_l \times N}$ is the channel matrix between the transmitter to the receiver $r_l$, and $\mathbf{z}_l \sim \mathscr{N}(\bm{0}, n_l \mathbf{I})$ represents the circularly symmetric Gaussian noise vector with covariance $n_l \mathbf{I}$ at $r_l$ (with $n_l \in \mathbb{R}_{++}$). The transmitted signal vector $\mathbf{x}$ is defined as $\mathbf{x} = \sum_{l=1}^L \mathbf{x}_l = \sum_{l=1}^L d_l \sqrt{p_l} \mathbf{u}_l$, where $\mathbf{u}_l \in \mathbb{C}^{N \times 1}$ is the normalized transmit beamformer, and $d_l$ and $p_l \in \mathbb{R}_{++}$ denote the information signal and transmit power for that specific stream, respectively. Suppose the decoding process for the $l$-th stream, utilizing a normalized receive beamformer, is represented by $\mathbf{v}_l \in \mathbb{C}^{N_l \times 1}$. The power vector is defined as $\mathbf{p} = \left[p_1, \ldots, p_L\right]^{\top}$, the transmit and receive beamforming matrices as $\mathbb{U} = \left(\mathbf{u}_1, \ldots, \mathbf{u}_L\right)$ and $\mathbb{V} = \left(\mathbf{v}_1, \ldots, \mathbf{v}_L\right)$ respectively, and the noise vector as $\mathbf{n} = \left[n_1, \ldots, n_L\right]^{\top}$. Additionally, we introduce the matrix $\mathbf{G} \in \mathbb{R}_{++}^{L \times L}$, where each entry is defined as $G_{l i} = \left|\mathbf{v}_l^{\dagger} \mathbf{H}_l \mathbf{u}_i\right|^2$. Here, we fix the transmit and receive beamformers $\mathbb{U}$ and $\mathbb{V}$. The Signal-to-Interference-plus-Noise Ratio (SINR) for the $l$-th stream in the downlink is expressed as \cite{tan2015wireless}
\begin{align}\label{eq:sinr_fixed}
    \operatorname{SINR}_l\left(\mathbf{p}\right) = \frac{G_{ll}p_l}{\sum_{j=1,j \neq l}^L G_{lj}p_j + n_l}.
\end{align}
}

Next, let us define a nonnegative matrix $\F$ with entries
\begin{equation}\label{eq:matrixF}
F_{lj}=\left\{\begin{matrix}
0, &\mbox{if}\;\; l = j \\ 
G_{lj}/G_{ll}, &\mbox{if}\;\; l \ne j
\end{matrix}\right..
\end{equation}
We assume that $\F$ is irreducible, meaning that each link has at least one interferer. Additionally, we define the constant vector
$$
\bm{\sigma}= \left(\frac{n_1}{G_{11}}, \frac{n_2}{G_{22}}, \dots,\frac{n_L}{G_{LL}}\right)^{T}.
$$

\textcolor{black}{
Our objective is to determine the power allocation vector $\mathbf{p}$ that maximizes the total transmitted data rate in the downlink, i.e., to solve the WSRM problem through power control, which is stated as follows \cite{tan2015wireless}:
\begin{align}\label{eq:sumrate_pro}
\textup{maximize} \;\;\; &\sum_{l=1}^{L}w_l\log(1+\textup{SINR}_l\left(\mathbf{p}\right)) \nonumber\\
\textup{subject to}  \;\;\;&\mathbf{m}^{\top}\p\leq \Bar{P},\nonumber \\
\;\;\;&\p\geq\bm{0},\\
\textup{variables:} \;\;\;&\p,\nonumber
\end{align}
where $\mathbf{w}=[w_1,\ldots,w_L]^{\top}> \bm{0}$ with $\mathbf{1}^{\top} \mathbf{w}=1$ is a given probability vector and $w_l$ reflects the $l$-th link with a larger weight indicating a higher priority, and  $\mathbf{m}=\left[m_1, \ldots, m_L\right]^{\top}$ is a weight vector with each $m_l \in \mathbb{R}_{++}$.}

\section{Optimization-theoretic-based Algorithm for downlink sum-rate maximization}\label{sec:BCD}

Note that the problem (\ref{eq:sumrate_pro}) is non-convex due to the non-convexity of the objective function. Specifically, the expression $\sum_{l=1}^{L} w_l \log(1 + \textup{SINR}_l(\mathbf{p})) = \sum_{l=1}^{L} w_l \log\left(\frac{\sum_{j=1}^L G_{lj} p_j + n_l}{\sum_{j=1, j \neq l}^L G_{lj} p_j + n_l}\right)$
represents the weighted sum logarithm of a signomial \cite{boyd2004convex}. \textcolor{black}{Transforming the signomial in this objective function into a convex function using techniques such as change-of-variables is challenging.} Consequently, the problem is inherently difficult to solve directly.

To address this problem, we first introduce a transformation $\p \mapsto \bm{\gamma}(\p)$, where
 \begin{align}\label{eq:rl}
\gamma_l(\p)=\textup{SINR}_l(\mathbf{p}),\; l=1,...,L.
\end{align}
Meanwhile, $\p$ can be written as a function of the rate vector:
\begin{align}
    \mathbf{p}(\bm{\gamma})=(\I-\operatorname{diag}(\boldsymbol{\gamma}) \F)^{-1} \operatorname{diag}(\boldsymbol{\gamma}) \bm{\sigma}.
\end{align}
Then the WSRM problem can be reformulated as follows \cite{tan2015wireless}:
\begin{align}\label{eq:sumrate_pro_r}
\textup{maximize} \;\;\; &\sum_{l=1}^{L}w_l\log(1+\gamma_l), \nonumber\\
\textup{subject to}  \;\;\;&\rho\left(\operatorname{diag}(\bm{\gamma})\left(\mathbf{F}+\left(1 / \Bar{P}\right) \bm{\sigma}\mathbf{m}^{\top}\right)\right) \leq 1, \\
\textup{variables:} \;\;\;&\bm{\gamma}\in \mathbb{R}^L_+.\nonumber
\end{align}
Setting $\Tilde{\bm{\gamma}} = \log \bm{\gamma}$,  we can further transform \eqref{eq:sumrate_pro_r} to the following problem \cite{tan2015wireless}:
\begin{align}\label{eq:sumrate_pro2}
\textup{maximize} \;\;\; &\sum_{l=1}^{L}w_l\log(1+e^{\Tilde{\gamma}_l}), \nonumber\\
\textup{subject to}  \;\;\;&\log\rho\left(\operatorname{diag}(e^{\Tilde{\bm{\gamma}}})\left(\mathbf{F}+\left(1 / \Bar{P}\right) \bm{\sigma}\mathbf{m}^{\top}\right)\right) \leq 0, \nonumber\\
\textup{variables:} \;\;\;&\Tilde{\bm{\gamma}}\in \mathbb{R}^L.
\end{align}

Due to the log-convexity property of the Perron-Frobenius eigenvalue \cite{kingman1961convexity}, $\log\rho\left(\operatorname{diag}(e^{\Tilde{\bm{\gamma}}})\left(\mathbf{F}+\left(1 / \Bar{P}\right) \bm{\sigma}\mathbf{m}^{\top}\right)\right)$ is a convex function in $\Tilde{\bm{\gamma}}$. However, the objective function remains nonconcave, posing a challenge in solving \eqref{eq:sumrate_pro2}. In the following discussion, we introduce an efficient BCD algorithm capable of converging to at least a local optimum and potentially reaching the global optimum when the initial point is appropriately chosen. While the proposed BCD algorithm may not consistently converge to the global optimum, its methodological efficiency and intrinsic analytical insights of \eqref{eq:sumrate_pro} offer a valuable pathway for integrating machine learning and domain-specific techniques. This integration provides new strategies with improved efficiency and better performance for addressing \eqref{eq:sumrate_pro}, as we will discuss in the next section.

Now we present the proposed BCD method for solving \eqref{eq:sumrate_pro2}. We begin by utilizing a well-known result that the log-sum-exp function is the convex conjugate of the negative entropy \cite{boyd2004convex}, which is expressed as follows:
\begin{align}\label{eq:entropy}
    \log \left(\sum_{n=1}^Ne^{x_n}\!\right)\!=\!\max_{\bm{y}}\!\left\{\bm{x}^T \bm{y}\!-\!\!\sum_{n=1}^N y_n \ln y_n\!:\! \bm{y} \!\geq\! \bm{0},\! \sum_{n=1}^N y_n\!=\!1\!\right\}.
\end{align}
By applying \eqref{eq:entropy} to the objective function in \eqref{eq:sumrate_pro2}, we can rewrite \eqref{eq:sumrate_pro2} as
\begin{align}\label{eq:sumrate_pro3}
\textup{maximize} \;\;\; &\sum_{l=1}^{L}w_l \left( \Tilde{\gamma}_ly_{l2}-\sum_{n=1}^2y_{ln}\log y_{ln}  \right) \nonumber\\
\textup{subject to}  \;\;\;&\log\rho\left(\operatorname{diag}(e^{\Tilde{\bm{\gamma}}})\left(\mathbf{F}+\left(1 / \Bar{P}\right) \bm{\sigma}\mathbf{m}^{\top}\right)\right) \leq 0, \nonumber\\
&\sum_{n=1}^2y_{ln}=1 \quad \forall l, \\
&y_{ln} \geq 0 \quad \forall l,n, \nonumber\\
\textup{variables:} \;\;\;&\Tilde{\gamma}_l, y_{ln},\;\forall\;l,n.\nonumber
\end{align}

Note that \eqref{eq:sumrate_pro3} exhibits convexity when considering the variable $\Tilde{\bm{\gamma}}$ or $\bm{y}$ individually. This observation motivates the utilization of the BCD method \cite{tseng2001convergence}, also known as the Gauss-Seidel method, to solve \eqref{eq:sumrate_pro3}.  In this approach, only one block of variables is optimized at each iteration while the remaining blocks are held constant. Specifically, at the $t$-th iteration of the BCD method for solving \eqref{eq:sumrate_pro3}, the block variable $\Tilde{\bm{\gamma}}^{(t+1)}$ is updated by solving the following subproblem:
\begin{align}\label{eq:sumrate_linear}
\textup{maximize} \;\;\; &\sum_{l=1}^{L}w_ly_{l2}^{(t)}\Tilde{\gamma}_l\nonumber\\
\textup{subject to}  \;\;\;&\log\rho\left(\operatorname{diag}(e^{\Tilde{\bm{\gamma}}})\left(\mathbf{F}+\left(1 / \Bar{P}\right) \bm{\sigma}\mathbf{m}^{\top}\right)\right) \leq 0, \\
\textup{variables:} \;\;\;&\Tilde{\gamma}_l, \;\forall\;l.\nonumber
\end{align}
Then the block variable $\bm{y}^{(t+1)}$ is updated by solving the following subproblem:
\begin{align}\label{eq:sumrate_entropy}
\textup{maximize} \;\;\; &\sum_{l=1}^{L}w_l \left( \Tilde{\gamma}^{(t+1)}_ly_{l2}-\sum_{n=1}^2y_{ln}\log y_{ln}  \right) \nonumber\\
\textup{subject to}  \;\;\;&\sum_{n=1}^2y_{ln}=1 \quad \forall l, \\
&y_{ln} \geq 0 \quad \forall l,n, \nonumber\\
\textup{variables:} \;\;\;&y_{ln},\;\forall\;l,n.\nonumber
\end{align}
It can be shown that the optimal solution of \eqref{eq:sumrate_entropy}, i.e., the updated  $y_{ln}^{(t+1)}$ is
\begin{align}\label{eq:solution_y}
   y_{ln}^{(t+1)} =\left\{\begin{matrix}
\frac{1}{1+e^{\Tilde{\gamma}^{(t+1)}_l}}, \quad n=1, \\ 
\frac{1}{1+e^{-\Tilde{\gamma}^{(t+1)}_l}},\quad  n=2.
\end{matrix}\right.
\end{align}

Then let us solve the problem \eqref{eq:sumrate_linear}. Notice that \eqref{eq:sumrate_linear} is convex, so one can use the interior-point method \cite{boyd2004convex} to obtain the optimal solution. However, we propose an alternative method, which is significantly faster for solving \eqref{eq:sumrate_linear}. The details are presented in Lemma \ref{thm:thm1}.
\begin{lemma}\label{thm:thm1}
     The optimal solution $\Tilde{\bm{\gamma}}^*$ of \eqref{eq:sumrate_linear}  can be expressed as follows:
     \begin{align}
         \Tilde{\gamma}_l^* = \log(z^*_l/(\B\mathbf{z^*})_l),
     \end{align}
     where  $\B = \mathbf{F}+\left(1 / \Bar{P}\right) \bm{\sigma}\mathbf{m}^{\top}$, and $z^*_l$ satisfies 
     \begin{align}\label{eq:fixed_point}
       \sum_i w_iy_{i2}^{(t)} \frac{b_{il}z_i^*}{\sum_j b_{ij}z_j^*} = w_ly_{l2}^{(t)}.
    \end{align}
   Moreover, $z^*_l$ can be obtained by the following iteration:
    \begin{align}\label{eq:iteration}
        z_l^{(k+1)} =  \frac{w_ly_{l2}^{(t)}}{\sum_i w_iy_{i2}^{(t)} \frac{b_{il}}{\sum_j b_{ij}z_j^{(k)}}}.
    \end{align}
Furthermore, the optimal Lagrangian dual $\lambda^*$ of \eqref{eq:sumrate_linear} is given by $\lambda^* = \sum_l w_ly_{l2}^{(t)}$.

\end{lemma}
\begin{IEEEproof}
    Please see the appendix for a detailed proof. 
\end{IEEEproof}

By iteratively solving \eqref{eq:sumrate_linear} using Lemma \ref{thm:thm1} and addressing \eqref{eq:sumrate_entropy} with \eqref{eq:solution_y}, we propose Algorithm \ref{alg:alg1}, an efficient BCD method  to solve \eqref{eq:sumrate_pro3}.

\begin{algorithm}
\SetAlgoLined
\textcolor{black}{
\KwIn{The system parameters $\G$, $\n,\mathbf{w},\Bar{P}$ and randomly initialized  $y_{ln}^{(0)}$ satisfying $\sum_{n=1}^2y_{ln}^{(0)}=1$ and $y_{ln}^{(0)}\geq 0$.}
\KwOut{The locally optimal solution $\p^{*}$ to \eqref{eq:sumrate_pro}.} 
}
\!1. Update $\Tilde{\bm{\gamma}}^{(t+1)}$:
\begin{align}\label{eq:gamma}
    \Tilde{\gamma}_l^{(t+1)} = \log(z_l^{*}/(\B\mathbf{z}^{*})_l),
\end{align}
where $z^{*}_l$ is the convergence of the following iteration with randomly initial $z_l^{(0)}$:
\begin{align}\label{eq:iteration_re}
        z_l^{(k+1)} =  \frac{w_ly_{l2}^{(t)}}{\sum_i w_iy_{i2}^{(t)} \frac{b_{il}}{\sum_j b_{ij}z_j^{(k)}}}.
    \end{align}
2. Update $y_{l2}^{(t+1)}$:
\begin{align}
    &y_{l2}^{(t+1)} = \frac{e^{\Tilde{\gamma}^{(t+1)}_l}}{1+e^{\Tilde{\gamma}^{(t+!)}_l}}.
\end{align}
3. Obtain optimal $\Tilde{\bm{\gamma}}^{*}$ by repeating Step 1 to Step 2 until convergence.\\
\textcolor{black}{
4. Obtain optimal power $\mathbf{p}^*$ with:
$$\mathbf{p}^*(\bm{\gamma}^*)=(\I-\operatorname{diag}(\boldsymbol{\gamma}^*) \F)^{-1} \operatorname{diag}(\boldsymbol{\gamma}^*) \bm{\sigma}.$$}
\caption{BCD for solving the WSRM problem}
\label{alg:alg1}
\end{algorithm}

\begin{lemma}\label{lem:lem2}
    Starting from any initial point $\bm{y}^{(0)}$, \textcolor{black}{$\Tilde{\bm{\gamma}}{(t)}$ in Algorithm \ref{alg:alg1} converges to at least one locally optimal solution $\Tilde{\bm{\gamma}}^{*}$}.
\end{lemma}
\begin{IEEEproof}
    Please see the appendix for a detailed proof. 
\end{IEEEproof}
\begin{remark}
Algorithm \ref{alg:alg1} is highly efficient as it updates the variable block $\bm{y}^{(t)}$ with a closed-form solution and updates $\Tilde{\bm{\gamma}}^{(t)}$ with a geometrically fast iterative algorithm. We present Example \ref{emp:emp1} to illustrate the effectiveness of Algorithm \ref{alg:alg1}.
\end{remark}

\begin{example}\label{emp:emp1}
Consider a wireless system with a single base station that serves $L=3$ users. The downlink transmit powers are denoted by $p_1$, $p_2$, and $p_3$. The given maximum power is $\Bar{P} = 2.50\text{ mW}$, and the noise power is $\bm{\sigma} = [0.01,0.01,0.01]^\top\text{ mW}$. We configure the channel gain matrix as $\G = [0.18, 0.05, 0.59; 0.61, 2.67, 0.56; 2.89, 0.99, 1.29]$.
The weight vectors are $\mathbf{w}=[0.33, 0.33, 0.34]^\top$ and $\mathbf{m}=[1, 1, 1]^\top$. Then we use Algorithm \ref{alg:alg1} with various randomly chosen initial power assignments to solve the WSRM problem. Fig. \ref{fig:fig1} illustrates the evolution of the data rate for the users. The results show that Algorithm \ref{alg:alg1} tends to converge to local optima in some cases, with the potential to reach the global optimum if the initial point is carefully chosen.

As mentioned above, the methodological efficiency and inherent analytical insights encapsulated in \eqref{eq:sumrate_pro3} by the proposed BCD present an avenue for the integration of machine learning and domain-specific techniques. In the subsequent section, we investigate combining DRL techniques with the proposed BCD method to propose a DRL-based BCD algorithm. This algorithm enables the integration of both optimization-theoretic-based and machine network-based techniques to solve \eqref{eq:sumrate_pro3} more efficiently and effectively.
\begin{figure}
\centerline{\includegraphics[scale=0.6]{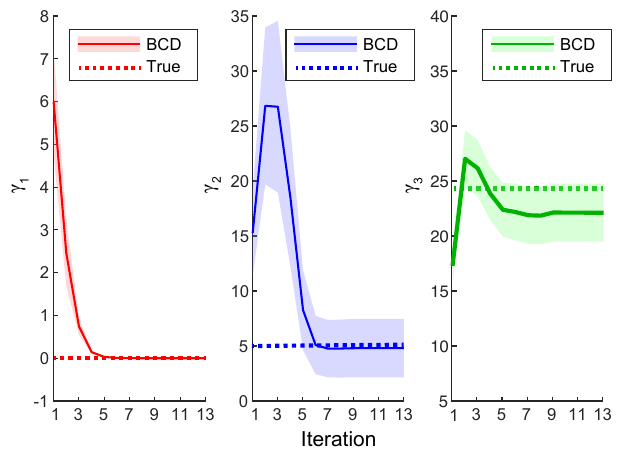}}
\captionsetup{font={footnotesize, stretch=1}, justification=raggedright}
\caption{Evolution of the data rate for $3$ users using the BCD method. \label{fig:fig1}}\vspace{-2mm}
\end{figure}

\end{example}


\section{DRL-based BCD for Downlink Sum-rate Maximization}\label{sec:RL-based BCD}
In this section, we present our DRL-based BCD algorithm, an interesting method that integrates the data-driven learning capabilities of DRL with the navigational and decision-making features of the optimization-theoretic-based BCD method. 

\begin{figure*}
\centerline{\includegraphics[scale=0.91]{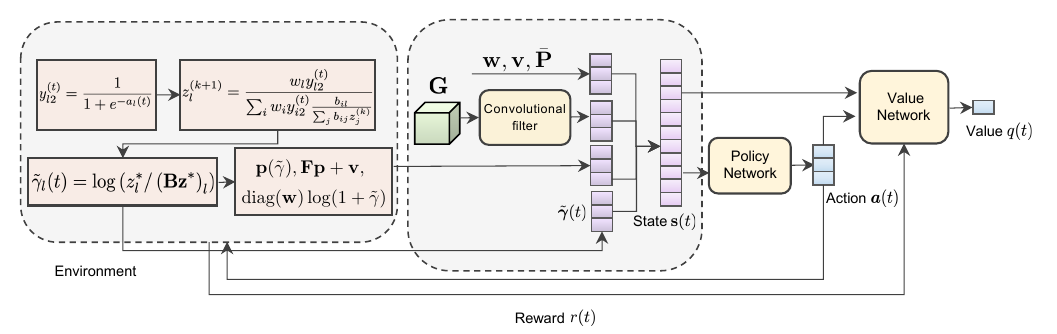}}
\captionsetup{font={footnotesize, stretch=1}, justification=raggedright}
\caption{The architecture of our DRL-based BCD algorithm for solving \eqref{eq:sumrate_pro3} includes state information from system parameters ($\G$, $\v$, $\mathbf{w}$, and $\Bar{P}$), the data rate $\Tilde{\bm{\gamma}}^{(t)}$ derived from the BCD algorithm for subproblems, and features of the problem's structure (including $F\p(\tilde{\bm{\gamma}})+\bm{\sigma}$ and $\textup{diag}(\mathbf{w})\log(1+\tilde{\gamma})$). The policy network outputs action $\bm{a}^{(t)}$ based on $\bm{s}^{(t)}$, and the agent receives reward $r^{(t)}$. Then the value network evaluates $\bm{a}^{(t)}$ in $\bm{s}^{(t)}$. This iteration continues until the agent converges to an optimal solution of \eqref{eq:sumrate_pro3}.  \label{fig:RL-BCD}}
\end{figure*}

\subsection{Overview of DRL for Continuous Control}\label{sec:overview}
DRL is a machine learning approach focused on training agents to make sequential decisions/actions in the environment to maximize cumulative rewards \cite{duan2016benchmarking}. Given that the optimal solution for \eqref{eq:sumrate_pro3} is achieved in a continuous space, we consider DRL methods for solving continuous control problems wherein the agent operates in continuous action and state spaces \cite{duan2016benchmarking,lillicrap2015continuous, fujimoto2018addressing,haarnoja2017reinforcement}.
One popular class of these DRL methods is the Actor-Critic framework \cite{lillicrap2015continuous, fujimoto2018addressing,haarnoja2017reinforcement}. This framework combines a learnable policy network (Actor) that learns a policy to take actions based on the observed state of the environment, and value networks (Critic) that estimate the expected future rewards to assess the actions.

At each time slot $t$, there are three crucial elements that the Actor-Critic framework considers: the state $\bm{s}^{(t)}$, the action $\bm{a}^{(t)}$, and the reward $r^{(t)}$. The state $\bm{s}^{(t)}$ provides the necessary environmental features for the agent to make decisions and take actions. The action $\bm{a}^{(t)}$ determines how the agent interacts with the environment. The reward $r^{(t)}$ is a feedback signal that the agent receives after performing action $\bm{a}^{(t)}$, which guides the learning process. Based on the given $\bm{s}^{(t)}$, the policy network, denoted as $\boldsymbol{\mu}\left(\boldsymbol{s}^{(t)}; \boldsymbol{\theta}\right)$ with the learning parameter $\boldsymbol{\theta}$, is typically implemented as a deep neural network model \cite{duan2016benchmarking}. It outputs a response to ${s}^{(t)}$ and, for continuous control tasks, directly leads to an action $\boldsymbol{a}^{(t)}$ instead of probabilities for potential actions.
The value network, referred to as $q(\bm{s}^{(t)},\bm{a}^{(t)};\bm{\omega})$ with learning parameter $\bm{\omega}$, is also a deep neural network model that approximates the action-value function \cite{duan2016benchmarking}. It produces a real-valued output that reflects the quality of the action $\bm{a}^{(t)}$ in $\bm{s}^{(t)}$. A higher value of $q(\bm{s}^{(t)},\bm{a}^{(t)};\bm{\omega})$ indicates a higher quality action. Therefore, the value network plays a critical role in accessing and refining the policy network. With the reward $r^{(t)}$ as feedback, the value network updates itself using the temporal difference (TD) method \cite{bertsekas2019reinforcement}.
During the training process, the value network assists in training the policy network. However, once the training is complete, the value network becomes unnecessary and can be discarded. Thus, during the decision-making process, the policy network takes control of the agent, guiding its actions based on the strategies it has learned.

Several efficient methods have been developed for training the policy and value networks in the Actor-Critic framework. These include the deep deterministic policy gradient (DDPG) \cite{lillicrap2015continuous}, the Twin Delayed Deep Deterministic policy gradient algorithm (TD3) \cite{fujimoto2018addressing}, and the soft Actor-Critic algorithm (SAC) \cite{haarnoja2017reinforcement}, etc. These methods typically update the policy network via gradient descent using Critic feedback, while the Critic updates value networks using the TD method. Additionally, techniques like target networks, policy noise, delayed and soft updates are commonly employed \cite{lillicrap2015continuous, fujimoto2018addressing, haarnoja2017reinforcement}.

\subsection{The DRL-based BCD model}
We now introduce our DRL-based BCD algorithm for solving \eqref{eq:sumrate_pro3}. 
The goal of our DRL-based BCD algorithm is to achieve the maximum sum-rate as fast as possible. We consider the BCD method as the agent and utilize the Actor-Critic framework of DRL to perform continuous control on the iteration updates (primarily on $\bm{y}(t)$ as discussed later) in the BCD method. To align this learning-based model with the WSRM problem, we present the three important elements discussed in section \ref{sec:overview} as follows:
\begin{itemize}
    \item State $\bm{s}^{(t)}$: 
    The state $\bm{s}^{(t)}$ consists of information from several components of the WSRM problem, which can be categorized into three primary parts. The first part includes the system parameters, such as $\G$, $\bm{\sigma}$, $\mathbf{w}$, and $\Bar{P}$. These parameters directly represent the system features of different problem instances and may change over time in practical applications. This aspect of state information has been considered in previous studies such as \cite{khan2020centralized}. The second part is the data rate $\Tilde{\bm{\gamma}}^{(t)}$, obtained from the $t$-th iteration of the BCD algorithm. This part represents one of our primary innovative contributions and holds greater significance as it leverages theoretical analysis of the problem's intrinsic structure. The third part of the state information incorporates elements such as $\p(\tilde{\bm{\gamma}})$, $F\p(\tilde{\bm{\gamma}})+\bm{\sigma}$, and $\textup{diag}(\mathbf{w})\log(1+\tilde{\gamma})$. This part provides a set of direct features of the problem's structure and has been explored in related studies \cite{rahmani2022deep,jin2025sum}. Thus, the state is designed as
    \begin{align}\label{eq:state}
        \bm{s}^{(t)}=\{&Conv\left(\G; \bm{\theta}_c\right),\bm{\sigma},\mathbf{w},\Bar{P},\Tilde{\bm{\gamma}}^{(t)},\nonumber\\
        &\p(\tilde{\bm{\gamma}}^{(t)}),\F\p(\tilde{\bm{\gamma}}^{(t)})+\bm{\sigma}, \textup{diag}(\mathbf{w})/(1+\tilde{\bm{\gamma}}^{(t)})\},
    \end{align}
    where $Conv\left(\;\cdot\;; \bm{\theta}_c\right)$ represents a convolutional filter \cite{goodfellow2016deep} and $\bm{\theta}_c$ is the set of learning parameters.     

    \;\;\;\;After executing the action $\bm{a}^{(t)}$, we update $\bm{s}^{(t)}$ to obtain $\bm{s}^{(t+1)}$. For a specific WSRM problem, the system parameters remain unchanged. The data rate $\Tilde{\bm{\gamma}}^{(t)}$ is updated in the following manner:
    \begin{align}\label{eq:iter1}
    \Tilde{\gamma}_l^{(t+1)} = \log(z_l^*/(\B\mathbf{z}^*)_l),
\end{align}
where $z_l^*$ is the convergence of the iteration:
\begin{align}\label{eq:iter2}
        z_l^{(k+1)} =  \displaystyle{\frac{w_ly_{l2}^{(t)}}{\sum_i w_iy_{i2}^{(t)} \frac{b_{il}}{\sum_j b_{ij}z_j^{(k)}}}},
    \end{align}
where $ y_{l2}^{(t)}$ a sigmoid function of the action $\bm{a}^{(t)}$, i.e.,
\begin{align}\label{eq:policy_y}
        y_{l2}^{(t)} = \frac{1}{1+e^{  -a_l^{(t)}}}.
    \end{align}
    \textcolor{black}{
    With the obtained $\Tilde{\gamma}_l^{(t+1)}$, the problem structure feature $\p(\tilde{\bm{\gamma}}^{(t)}),\F\p(\tilde{\bm{\gamma}}^{(t)})+\bm{\sigma} $ and $ \textup{diag}(\mathbf{w})/(1+\tilde{\bm{\gamma}}^{(t)})$ can be updated accordingly.}

    \;\; The inclusion of $\Tilde{\bm{\gamma}}^{(t)}$ in $\bm{s}^{(t)}$  is, in fact, one of the most important steps in integrating DRL and BCD in this paper. It enables the environment to utilize the analytical structure of the WSRM problem, specifically the iteration of $\Tilde{\bm{\gamma}}^{(t)}$ in BCD, in response to the agent's actions. This facilitates the transmission of these analytical structures to the policy network. Moreover, the iteration of $\Tilde{\bm{\gamma}}^{(t)}$ converges exponentially fast, enabling the environment to deliver timely responses promptly.
    \item Action $\bm{a}^{(t)}$: As shown in \eqref{eq:policy_y}, we use the policy network to determine the action for $\bm{y}^{(t)}$ in \eqref{eq:sumrate_pro3}. This differs from most existing works such as \cite{waraiet2023robust}, which typically focus on the action on the power $\p$ in \eqref{eq:sumrate_pro}. In fact, training the action on $\bm{y}^{(t)}$ greatly improves accuracy and reduces training difficulties in our DRL-based BCD algorithm. That is because our algorithm only needs to find favorable starting points for the iteration of the BCD method, which typically forms a continuous set. As long as the policy network's output falls within this set, our algorithm can effectively reach a globally optimal solution. However, when training actions on $\p$, it is challenging to approximate the globally optimal solution with absolute precision.  Given that $\bm{y}^{(t)}$ is a continuous variable, we design the policy network $\bm{\mu} (\bm{s}^{(t)};\bm{\theta})$ for continuous control. For a given state $\bm{s}^{(t)}$, the action $\bm{a}^{(t)}$ is obtained by
    \textcolor{black}{
    \begin{align}\label{eq:policy}
    \bm{a}^{(t)} &= \bm{\mu} (\bm{s}^{(t)};\bm{\theta}) + \bm{\sigma} \nonumber\\
    & = \bm{\mathcal{F}}_D\left(\cdots\bm{\mathcal{F}}_2\left(\bm{\mathcal{F}}_1\left(\bm{s}^{(t)};\bm{\theta}_1\right);\bm{\theta}_2\right)\cdots;\bm{\theta}_D\right)+ \bm{\sigma},
\end{align}
where $\bm{\theta}=\left\{\bm{\theta}_c,\bm{\theta}_1,\bm{\theta}_2, \cdots,\bm{\theta}_D\right\}$ is the set of learning parameters. The function $\bm{\mathcal{F}}_d(\cdot;\bm{\theta}_d)$ represents a compositing function that consists of linear combinations and activation functions (e.g., the ReLU function) with parameter $\bm{\theta}_d$ for the $d$-th hidden layer in a fully connected deep neural network model.} Additionally, $\bm{\sigma}$ is the noise that follows a continuous uniform distribution, i.e., $\bm{\sigma} \sim \bm{U}(0, c)$ with $c$ as an adapted positive hyperparameter.
    \item Reward $r^{(t)}$: We reward DRL-based BCD algorithms that converge quickly to the optimal solution $\Tilde{\bm{\gamma}}^*$ of \eqref{eq:sumrate_pro3} and penalize those that do not. A natural penalty would be the cumulative differences between the predicted objective values and the optimal objective value over all iterations. Then at time $t$, the reward is defined as
    \begin{align}\label{eq:reward}
       r^{(t)}(\bm{s}^{(t)})= -\sum_{l=1}^{L}w_l\left(\log\left(\frac{1+e^{\Tilde{\gamma}^{(t)}_l}}{1+e^{\Tilde{\gamma}^*_l}}\right)\right)^2.
\end{align}
    The absolute value of the cumulative reward is larger when the DRL-based BCD algorithm converges to the optimal solution of \eqref{eq:sumrate_pro3} slowly and is smaller otherwise. 
\end{itemize}

To expedite the DRL-based BCD algorithm in achieving the maximum sum-rate quickly, we need to maximize the expected cumulative reward. Specifically, for a trajectory with $T$ steps, the goal is to maximize the following $J(\bm{\theta})$:

\begin{align}
    J(\bm{\theta})\!\!=\!\!\sum_{\bm{s}^{(t)}\in \bm{\mathcal{S}}}\!\!\mathcal{p}(\bm{s}^{(0)})\prod_{t=0}^{T-1}\mathcal{p}(\bm{s}^{(t+1)}|\bm{\mu} (\bm{s}^{(t)};\bm{\theta}) )\sum_{t=0}^{T-1}\beta^tr^{(t)}(\bm{s}^{(t)}),
\end{align}
where $\bm{s}^{(0)}$ is the randomly sampled started state with probability $\mathcal{p}(\bm{s}^{(0)})$, $\bm{\mathcal{S}}$ represents the finite-size set of possible $\bm{s}^{(t)}$ under the assumption of a finite number of local optima for \eqref{eq:sumrate_pro3}, $\mathcal{p}(\bm{s}^{(t+1)}|\bm{\mu} (\bm{s}^{(t)};\bm{\theta})$ is the state transition probability to $\bm{s}^{(t+1)}$ from  $\bm{s}^{(t)}$ under policy $\bm{\mu} (\bm{s}^{(t)};\bm{\theta})$, and $\beta$ is a discount factor. 
As shown in Fig. \ref{fig:RL-BCD}, the training process proceeds as follows: \textcolor{black}{At time slot $t$, the state $\bm{s}^{(t)}$ concatenates the system parameters and environment-provided data rate  $\Tilde{\bm{\gamma}}^{(t)}$}. Based on $\bm{s}^{(t)}$, the agent performs an action $\bm{a}^{(t)}$ using the policy network in  \eqref{eq:policy} and receives a reward $r^{(t)}$. The action $\bm{a}^{(t)}$ is then assessed by a value network $q(\bm{s}^{(t)},\bm{a}^{(t)};\bm{\omega})$, which is a fully connected deep neural network with learning parameters $\bm{\omega}$, defined as follows:
\textcolor{black}{
\begin{align}\label{eq:value-net}
    &q(\bm{s}^{(t)},\bm{a}^{(t)};\bm{\omega}) \nonumber\\
=&\bm{\mathcal{F}}_D\left(\cdots\bm{\mathcal{F}}_2\left(\bm{\mathcal{F}}_1\left(\bm{s}^{(t)},\bm{a}^{(t)};\bm{\omega}_1\right);\bm{\omega}_2\right)\cdots;\bm{\omega}_D\right),
\end{align}
where $\bm{\omega}=\left\{\bm{\omega}_1,\bm{\omega}_2, \cdots,\bm{\omega}_D\right\}$ is the set of learning parameters, and the function $\bm{\mathcal{F}}_d(\cdot;\bm{\omega}_d)$ also represents a compositing function that consists of linear combinations and activation functions (e.g., the ReLU function) with parameter $\bm{\omega}_d$ for the $d$-th hidden layer in a fully connected deep neural network model.}
Then the policy network is refined according to the assessment, and simultaneously the value network updates itself based on the reward $r^{(t)}$. More information on this updating process for training the policy and value networks is provided in the following subsection. Then the environment responds to  $\bm{a}^{(t)}$ by updating $\Tilde{\bm{\gamma}}^{(t)}$ to $\Tilde{\bm{\gamma}}^{(t+1)}$ through convergence iterations \eqref{eq:iter1} to \eqref{eq:iter2}, thereby updating the state to $\bm{s}^{(t+1)}$. With $\bm{s}^{(t+1)}$, the agent generates and executes a new action, and then both the policy and value networks are concurrently updated. This iterative cycle continues until the reward $r^{(t)}$ approaches a tolerable error $\varepsilon$ or the maximum iteration threshold is reached. The whole detailed procedure is outlined in Algorithm \ref{alg:alg2}. Moreover, in Appendix B, we show the convergence of the DRL-based BCD in Algorithm \ref{alg:alg2}.

\begin{algorithm}
\SetAlgoLined
\KwIn{The system parameters $\G$, $\n,\mathbf{w},\Bar{P}$ and randomly initialized   $\Tilde{\bm{\gamma}}(0)$.}
\KwOut{The  optimal parameter $\bm{\theta}^*$ and $\bm{\omega}^*$.} 
1. Predict the action $\bm{a}^{(t)}$ with policy network $\bm{\mu} (\bm{s}^{(t)};\bm{\theta})$ in \eqref{eq:policy} and assess $\bm{a}^{(t)}$ with the value network $q(\bm{s}^{(t)},\bm{a}^{(t)};\bm{\omega})$ in \eqref{eq:value-net}.\\
2. Update the state $\bm{s}^{(t+1)}$ with $\Tilde{\bm{\gamma}}^{(t+1)}$:
\begin{align*}
    \Tilde{\gamma}_l^{(t+1)} = \log(z_l^{(t+1)}/(\B\mathbf{z}^{(t+1)})_l),
\end{align*}
where $z_l^{(t+1)}$ is the convergence of the iteration:
\begin{align*}
        z_l^{(k+1)} =  \frac{w_ly_{l2}^{(t)}}{\sum_i w_iy_{i2}^{(t)} \frac{b_{il}}{\sum_j b_{ij}z_j^{(k)}}},
    \end{align*}
and $ y_{l2}^{(t)}$  responses to the action $\bm{a}^{(t)}$:
\begin{align*}
        y_{l2}^{(t)} = \frac{e^{ a_l(\bm{s}^{(t)},\bm{\theta})}}{1+e^{ a_l(\bm{s}^{(t)},\bm{\theta})}}.
    \end{align*}

3. Update the reward $r^{(t)}$:
\begin{align*}
       r^{(t)}= -\sum_{l=1}^{L}w_l\left(\log\left(\frac{1+e^{\Tilde{\gamma}^{(t+1)}_l}}{1+e^{\Tilde{\gamma}^*_l}}\right)\right)^2.
\end{align*}
\!\!\!4. Predict the action with the policy network $\bm{\mu} (\bm{s}^{(t+1)};\bm{\theta})$ and assess the action with the value network $\bm{q} (\bm{s}^{(t+1)},{\bm{a}}^{(t+1)};\bm{w})$.\\
5. Update the policy network $\bm{\mu} (\bm{s}^{(t)};\bm{\theta})$ and the value network $\bm{q} (\bm{s}^{(t)},{\bm{a}}^{(t)};\bm{w})$, as presented in Section \ref{sec:trainign}.\\
6. Repeat Step 1-5 until the reward $r^{(t)}$ converges to a tolerable error $\varepsilon$.
\caption{DRL-based BCD for Downlink WSRM}
\label{alg:alg2}
\end{algorithm}

\begin{remark}
    The main features of Algorithm \ref{alg:alg2} lie in the state transition (by incorporating the iterations \eqref{eq:iter1}-\eqref{eq:iter2} to obtain $\Tilde{\bm{\gamma}}^{(t)}$ in $\bm{s}^{(t)}$) and the agent's actions on $\bm{y}^{(t)}$ (to update the construction of subproblems \eqref{eq:sumrate_linear} of \eqref{eq:sumrate_pro3}). 
    The latter enhances the robustness of the algorithm. 
    Regarding the former, based on the iterations in \eqref{eq:iter1}-\eqref{eq:iter2} to reach the optimal solution for subproblem \eqref{eq:sumrate_linear} of \eqref{eq:sumrate_pro3}, the output, $\Tilde{\bm{\gamma}}^*$ obtained from the DRL-based BCD in the decision stage, not only meets the conditions in \eqref{eq:sumrate_pro3} but also has the potential to achieve a $100\%$ accurate optimal solution.
    Essentially, the DRL-based BCD algorithm leverages the learning capabilities of DRL to obtain superior initial points and overcome the trap of local optima, facilitating BCD's convergence towards the global optimum.
\end{remark}

\subsection{Training Process}\label{sec:trainign}
To train the policy and value network in the DRL-based BCD algorithm, we employ the training techniques of DDPG and TD3 and refer to the model as the DDPG-based BCD algorithm and TD3-based BCD algorithm, respectively.

\begin{itemize}
    \item DDPG-based BCD: In Step 4 of Algorithm \ref{alg:alg2}, we integrate various DDPG techniques \cite{lillicrap2015continuous}, including the replay buffer, independent target networks $\hat{\bm{\mu}} (\bm{s}^{(t)};\hat{\bm{\theta}})$ and $\hat{q}(\bm{s}^{(t)},\bm{a}^{(t)};\hat{\bm{\omega}})$, noise-based exploration with  $\bm{\sigma}\sim \bm{U}(0, c)$, and soft parameter updates with $\tau\in(0,1)$, to train the policy and value network. The replay buffer diversifies training data and reduces variance. Using target networks minimizes sudden weight changes and prevents oscillations. Adding noise to the policy promotes exploration and avoids local optima. Soft parameter updates prevent instability and improve training stability for more efficient learning. 
When sampling a random minibatch of $N$ transitions $\left(\bm{s}^{(t)}, \bm{a}^{(t)}, r^{(t)}, \bm{s}^{(t+1)}\right)$ obtained by Step 1-3 in Algorithm \ref{alg:alg2}, let $\mathcal{V}^{(t)}=r^{(t)}+\gamma \hat{q}\left(\bm{s}^{(t+1)}, \hat{\mu}\left(\bm{s}^{(t+1)} ; \hat{\theta}\right) ; \hat{\bm{\omega}}\right)$. Then the value network can be updated by minimizing the loss: 
\begin{align}\label{eq:update_value}
    \mathcal{L}=\frac{1}{N} \sum_i\left(\mathcal{V}^{(t)}-q\left(\bm{s}^{(t)}, \bm{a}^{(t)} ; \bm{\omega}\right)\right)^2.
\end{align}
For updating the policy network, the sampled policy gradient can be used, i.e.,
\begin{align}\label{eq:update_policy}
    \nabla_{\bm{\theta}} J \approx \frac{1}{N} \sum_i \nabla_{\bm{a}} q\left(\bm{s}^{(t)}, \bm{a}^{(t)} ; \bm{\omega}\right) \nabla_{\bm{\theta}} \mu\left(\bm{s}^{(t)} ; \bm{\theta}\right).
\end{align}
Then the target networks are updated by
$$
\begin{aligned}
\hat{\bm{\omega}} & \leftarrow \tau \bm{\theta}+(1-\tau) \hat{\bm{\omega}}, \\
\hat{\theta} & \leftarrow \tau \bm{\theta}+(1-\tau) \hat{\theta}.
\end{aligned}
$$

\item TD3-based BCD: Besides DDPG, we also use the TD3 technique to train the policy network and value network in Step 4 of Algorithm \ref{alg:alg2}. Similarly, based on the transitions $\left(\bm{s}^{(t)}, \bm{a}^{(t)}, r^{(t)}, \bm{s}^{(t+1)}\right)$, TD3 utilizes the TD method to update the value network and employs the feedback from value network to update the policy network through gradient descent, as shown in \eqref{eq:update_value} and \eqref{eq:update_policy}. To address the problem of overestimation, TD3 introduces two sets of value networks to mitigate estimation bias. Meanwhile, both the value networks and the policy network maintain their target networks, ensuring stable updates during training. Moreover, the policy networks are updated less frequently than the value networks, which improves the stability of the critic by preventing inaccurate estimations from adversely affecting policy updates. 
    
\end{itemize}


\subsection{Complexity Analysis}
The computational complexity of the DRL-based BCD algorithm includes both the model training and prediction stages. During training, a substantial amount of data is input into the DRL algorithm until the model achieves specific performance criteria. In the prediction stage, the trained policy network is used for real-time predictions, involving action updates and state transitions. Assuming there are $L$ users in the wireless network and the convolutional filter has dimensions of $K\times K$, the computational complexity of the convolutional layer is $O((L-K)^2\times K^2)$. Subsequently, the computational complexity of the policy network is $O((L-K)^2\times K^2+L)$. According to Lemma \ref{thm:thm1}, the complexity of updating the state is $O(\log L)$. Therefore, the computational complexity of updating the actions grows to $O((L-K)^2\times K^2+L+\log L)$ per time as the wireless network size increases.
In the model training stage, the computational complexity involves updating the parameters of the policy and value networks, as well as state transitions. Assuming there are $m$ parameters in the policy network and $n$ parameters in the value network, and considering the action dimension as $L$ with a state update complexity of $O(\log L)$, the update complexities for the DRL-based BCD model become $O(mL+\log L+n)$ per time as the wireless network size expands.

\section{Application Extension for Downlink WSRM
Beamforming}\label{sec:extend}
We now demonstrate the strong extensibility of the DRL-based BCD framework by applying it to the scenario of joint beamforming for downlink WSRM. Specifically, the $\mathbb{U}$ and $\mathbb{V}$ are not fixed in \eqref{eq:sinr_fixed}; rather, they are joint variables to be optimized for maximizing the downlink WSRM problem. We first consider the case where the transmit beamformer $\mathbb{U}$ is not fixed and the receive beamformer $\mathbb{V}$ is fixed. Then the downlink WSRM problem is given by
\begin{align}\label{eq:sumrate_beam}
\textup{maximize} \;\;\; &\sum_{l=1}^{L}w_l\log(1+\textup{SINR}_l\left(\mathbf{p}, \mathbb{U}\right)) \nonumber\\
\textup{subject to}  \;\;\;&\mathbf{m}^{\top}\p\leq \Bar{P},\nonumber \\
\;\;\;&\p\geq \bm{0},\nonumber\\
\;\;\;&\|\mathbf{u}_l\|=1, \; \forall l,\\
\textup{variables:} \;\;\;&\p, \mathbb{U}.\nonumber
\end{align}
The virtual dual network of this system comprises the channels that are the transposes of the channels in the primal network. The transmit beamforming vectors become the dual receive beamforming vectors. 
Then the SINR of the virtual dual network takes the form of
\begin{align}\label{eq:sinr2}
    \operatorname{SINR}'_l\left(\mathbf{p'},\mathbb{U}\right) &=\frac{\left|\mathbf{u}_l^{\dagger} \mathbf{H}^{\dagger}_l \mathbf{v}_l\right|^2p'_l}{\sum_{j=1,j \neq l}^L \left|\mathbf{u}_l^{\dagger} \mathbf{H}^{\dagger}_j \mathbf{v}_j\right|^2p'_j + n'_l}\\\nonumber
    &= \frac{G_{ll}p'_l}{\sum_{j=1,j \neq l}^L G_{jl}p'_j + n'_l}.
\end{align}
The weighted WSRM problem of the virtual dual network is
\begin{align}\label{eq:sumrate_dual}
\textup{maximize} \;\;\; &\sum_{l=1}^{L}w_l\log(1+\textup{SINR}'_l\left(\mathbf{p}', \mathbb{U}\right)) \nonumber\\
\textup{subject to}  \;\;\;&\mathbf{m}'^{\top}\p'\leq \Bar{P}',\nonumber \\
\;\;\;&\p'\geq \bm{0},\nonumber\\
\;\;\;&\|\mathbf{u}_l\|=1, \; \forall l,\\
\textup{variables:} \;\;\;&\p', \mathbb{U}.\nonumber
\end{align}
Then with the uplink-downlink duality mapping $\{\mathbf{m}'^{\top},\n',\Bar{P}'\} \leftrightarrow \{\mathbf{n}^{\top},\mathbf{m}^{\top},\Bar{P}\}$, the downlink WSRM problem \eqref{eq:sumrate_beam} is equivalent to \eqref{eq:sumrate_dual}. This result has been shown in Theorem 2 in \cite{lai2014beamforming}. Indeed, it can also be seen from that the constraint $\rho\left(\operatorname{diag}(\bm{\gamma})\left(\mathbf{F}+\left(1 / \Bar{P}\right) \bm{\sigma}\mathbf{m}^{\top}\right)\right) \leq 1$ is equivalent to $\rho\left(\operatorname{diag}(\bm{\gamma})\left(\mathbf{F}^{\top}+\left(1 / \Bar{P}\right) \mathbf{m}\bm{\sigma}^{\top}\right)\right) \leq 1$ due to the fact that $\rho(\mathbf{X}\mathbf{Y})=\rho(\mathbf{Y}\mathbf{X})$ and $\rho(\mathbf{X}^{\top})=\rho(\mathbf{X})$. With the duality mapping and the same analysis process as in Section \ref{sec:BCD}, we have the following lemma.
\begin{lemma}\label{lem:lem3}
    With the uplink-downlink duality mapping $\{\mathbf{m}'^{\top},\n',\Bar{P}'\} \leftrightarrow \{\mathbf{n}^{\top},\mathbf{m}^{\top},\Bar{P}\}$ and the given $\mathbf{u}^*_l$ in \eqref{eq:sumrate_dual}, the BCD algorithm in Algorithm 1 also solves \eqref{eq:sumrate_dual}, with the matrix $\B$ being substituted by $\B'=\mathbf{F}^{\top}+\left(1 / \Bar{P}\right) \mathbf{m}\bm{\sigma}^{\top}$.
\end{lemma}

Meanwhile, since the $\operatorname{SINR}'_l\left(\mathbf{p'},\mathbb{U}\right) $ in \eqref{eq:sinr2} is a function of $\mathbb{U}$ but only depends on $\mathbf{u}_l$, the optimization of the receive beamforming in the virtual dual network is decoupled from one another. Then the $l$-th receive beamforming is optimized simply by maximizing its rate $r_l$, which is equivalent to the maximization of $\operatorname{SINR} \gamma_l^{\prime}$ and the minimization of the mean-squared error (MSE), leading to the LMMSE beamforming vector as the optimal solution, i.e.,
\begin{align}\label{eq:beamforming}
    \mathbf{u}_l^{\dagger \star}(\p'^*)=\frac{\left(\sum_{j \neq l} p'^*_{j} \mathbf{H}_j^{\dagger} \mathbf{v}_j \mathbf{v}_j^{\dagger} \mathbf{H}_j+m_l \mathbf{I}\right)^{-1}\mathbf{H}_l^{\dagger} \mathbf{v}_l}{\left\|\left(\sum_{j \neq l} p'^*_{j} \mathbf{H}_j^{\dagger} \mathbf{v}_j \mathbf{v}_j^{\dagger} \mathbf{H}_j+m_l \mathbf{I}\right)^{-1}\mathbf{H}_l^{\dagger} \mathbf{v}_l\right\|}.
\end{align}

When the receive beamformer $\mathbb{V}$ is not fixed and the transmit beamformer $\mathbb{U}$ is fixed. Then the downlink WSRM problem can be expressed as
\begin{align}\label{eq:sumrate_beam2}
\textup{maximize} \;\;\; &\sum_{l=1}^{L}w_l\log(1+\textup{SINR}_l\left(\mathbf{p}, \mathbb{V}\right)) \nonumber\\
\textup{subject to}  \;\;\;&\mathbf{m}^{\top}\p\leq \Bar{P},\nonumber \\
\;\;\;&\p\geq \bm{0},\nonumber\\
\;\;\;&\|\mathbf{v}_l\|=1, \; \forall l,\\
\textup{variables:} \;\;\;&\p, \mathbb{V}.\nonumber
\end{align}

Similarly, since the $\operatorname{SINR}_l\left(\mathbf{p},\mathbb{V}\right)$ is a function of $\mathbb{V}$ but also depends solely on $\mathbf{v}_l$, the optimal solution can be expressed as the LMMSE beamforming vector as
\begin{align}\label{eq:beamforming_v}
    &\mathbf{v}_l^{\dagger}(\p^*)\nonumber\\
    =&\frac{\left(\sum_{j \neq l} p_{j}\mathbf{H}_l\mathbf{u}_j  \mathbf{u}_j^{\dagger} \mathbf{H}_j^{\dagger}+\sigma_l \mathbf{I}\right)^{-1}\mathbf{H}_l \mathbf{u}_l}{\left\|\left(\sum_{j \neq l} p_{j}\mathbf{H}_l\mathbf{u}_j \mathbf{u}_j^{\dagger} \mathbf{H}^{\dagger}_l+\sigma_l \mathbf{I}\right)^{-1}\mathbf{H}_l\mathbf{u}_l\right\|}.
\end{align}

By integrating the framework of the DRL-based BCD as outlined in Algorithm \ref{alg:alg2}, along with Lemma \ref{lem:lem3}, \eqref{eq:beamforming} and \eqref{eq:beamforming_v}, we can readily apply the DRL-based BCD model to tackle the downlink WSRM beamforming. The architecture is presented in Algorithm \ref{alg:alg3}. While the configuration of the policy network and value network remains the same structures as in Algorithm \ref{alg:alg2}, the state $\bm{s}^{(t)}$ is updated as

\begin{align}\label{eq:state2}
     \bm{s}^{(t+1)}=\Big\{&Conv\left(\mathbf{G}^{(t+1)}; \bm{\theta}_c\right),\bm{\sigma},\mathbf{w},\Bar{P},\nonumber\\
     &\Tilde{\bm{\gamma}}^{(t+1)},\mathbb{U}(\bm{\Tilde{\gamma}'^{(t+1)}}),\mathbb{V}(\bm{\Tilde{\gamma}^{(t+1)}}),\nonumber\\
&\p(\tilde{\bm{\gamma}}^{(t+1)}),\F\p(\tilde{\bm{\gamma}}^{(t+1)})+\bm{\sigma}, \textup{diag}(\mathbf{w})/(1+\tilde{\bm{\gamma}}^{(t+1)})\Big\}.
\end{align}

In the state $\bm{s}^{(t+1)}$, we have
\begin{align*}
    {\Tilde{\gamma}}_l'^{(t+1)} = \log({z'}_l^{(t+1)}/(\B'^{(t+1)}\mathbf{z}'^{(t+1)})_l),
\end{align*}
in which $\B'^{(t)}\!=\!\F^{^{(t)}\top}\!+\!\frac{1}{\bar{P}}\mathbf{m}\bm{\sigma}^{\top}$ with $F^{(t)}_{lj}\!=\!|\mathbf{v}_l^{\dagger(t)}\mathbf{H}_l\mathbf{u}^{(t)}_j|^2 /|\mathbf{v}_l^{\dagger(t)}\mathbf{H}_l \mathbf{u}^{(t)}_l|^2$ for $j\neq l$ and $F^{(t+1)}_{ll}\!=0$, and  ${z'}_l^{(t+1)}$ is the convergence of the iteration:
\begin{align*}
        {z'}_l^{(k+1)} =  \displaystyle{\frac{w_ly_{l2}^{(t)}}{\sum_i w_iy_{i2}^{(t)} \frac{{b'^{(t)}}_{il}}{\sum_j {b'}^{(t)}_{ij}{z'}_j^{(k)}}}},
    \end{align*}
where $ y_{l2}^{(t)}$  responses to the action $\bm{a}^{(t)}$:
\begin{align*}
        y_{l2}^{(t)} =\displaystyle{ \frac{e^{ a_l(\bm{s}^{(t)},\bm{\theta})}}{1+e^{ a_l(\bm{s}^{(t)},\bm{\theta})}}}.
    \end{align*}
The transmit beamforming vector is updated as
\begin{align*}
    &\mathbf{u}_l(\bm{\Tilde{\gamma}'^{(t+1)}})\\
    =&\frac{\left(\sum_{j \neq l} p'_{j}(\bm{\Tilde{\gamma}'^{(t+1)}}) \mathbf{H}_j^{\dagger} \mathbf{v}^{(t)}_j \mathbf{v}^{\dagger(t)}_j \mathbf{H}_j+m_l \mathbf{I}\right)^{-1}\mathbf{H}_l^{\dagger} \mathbf{v}^{(t)}_l}{\left\|\left(\sum_{j \neq l} p'_{j}(\bm{\Tilde{\gamma}'^{(t+1)}}) \mathbf{H}_j^{\dagger} \mathbf{v}^{(t)}_j \mathbf{v}^{\dagger(t)}_j \mathbf{H}_j+m_l \mathbf{I}\right)^{-1}\mathbf{H}_l^{\dagger} \mathbf{v}^{(t)}_l\right\|}.
\end{align*}

Also, we have that
\begin{align*}
    {\Tilde{\gamma}}_l^{(t+1)} = \log({z}_l^{(t+1)}/(\B^{(t+1)}\mathbf{z}^{(t+1)})_l),
\end{align*}
with $\B^{(t+1)}\!=\!\F^{(t+1)}\!+\!\frac{1}{\bar{P}}\bm{\sigma}\mathbf{m}^{\top}$ where $F^{(t+1)}_{lj}\!=\!|\mathbf{v}^{\dagger(t)}_l\mathbf{H}_l\mathbf{u}^{(t+1)}_j| ^2/|\mathbf{v}^{\dagger(t)}_l\mathbf{H}_l\mathbf{u}^{(t+1)}_l|^2$ for $j\neq l$ and $F_{ll}^{(t+1)}=0$,  and  ${z}_l^{(t+1)}$ is the convergence of the iteration
\begin{align*}
        {z}_l^{(k+1)} =  \displaystyle{\frac{w_ly_{l2}^{(t)}}{\sum_i w_iy_{i2}^{(t)} \frac{{b}^{(t+1)}_{il}}{\sum_j {b}^{(t+1)}_{ij}{z}_j^{(k)}}}},
    \end{align*}
and the receive beamforming vector is updated as
\begin{align*}
    &\mathbf{v}_l(\bm{\Tilde{\gamma}^{(t+1)}})\\
    =&\frac{\left(\sum_{j \neq l} p_{j}(\bm{\Tilde{\gamma}^{(t+1)}}) \mathbf{H}_l\mathbf{u}^{(t+1)}_j \mathbf{u}^{\dagger(t+1)}_j \mathbf{H}_j^{\dagger}+\sigma_l \mathbf{I}\right)^{-1}\mathbf{H}_l \mathbf{u}^{(t+1)}_l}{\left\|\left(\sum_{j \neq l} p_{j}(\bm{\Tilde{\gamma}^{(t+1)}}) \mathbf{H}_j\mathbf{u}^{(t+1)}_j \mathbf{u}^{\dagger(t+1)}_j \mathbf{H}^{\dagger}_j+\sigma_l \mathbf{I}\right)^{-1}\mathbf{H}_l\mathbf{u}^{(t+1)}_l\right\|}.
\end{align*}

\begin{algorithm}
\SetAlgoLined
\KwIn{The system parameters $\G$, $\n,\mathbf{w},\Bar{P}$ and randomly initialized   $\Tilde{\bm{\gamma}}(0)$.}
\KwOut{The  optimal parameter $\bm{\theta}^*$ and $\bm{\omega}^*$.} 
1. Predict the action $\bm{a}^{(t)}$ with policy network $\bm{\mu} (\bm{s}^{(t)};\bm{\theta})$ in \eqref{eq:policy} and assess $\bm{a}^{(t)}$ with the value network $q(\bm{s}^{(t)},\bm{a}^{(t)};\bm{\omega})$ in \eqref{eq:value-net}.\\
2. Update the state $\bm{s}^{(t+1)}$ as:
\begin{align*}
     &\bm{s}^{(t+1)}=\{Conv\left(\mathbf{H}; \bm{\theta}_c\right),\bm{\sigma},\mathbf{w},\Bar{P},\Tilde{\bm{\gamma}}^{(t)},\mathbb{U}(\bm{\Tilde{\gamma}'^{(t+1)}}),\\
    & \mathbb{V}(\bm{\Tilde{\gamma}^{(t+1)}}),
\p(\tilde{\bm{\gamma}}^{(t)}),\F\p(\tilde{\bm{\gamma}}^{(t)})\!\!+\!\!\bm{\sigma}, \textup{diag}(\mathbf{w})\log(1\!+\!\tilde{\bm{\gamma}}^{(t)})\},
\end{align*}
3. Update the reward $r^{(t)}$:
\begin{align*}
       r^{(t)}= -\sum_{l=1}^{L}w_l\left(\log\left(\frac{1+e^{\Tilde{\gamma}^{(t+1)}_l}}{1+e^{\Tilde{\gamma}^*_l}}\right)\right)^2.
\end{align*}
\!\!\!4. Predict the action with the policy network $\bm{\mu} (\bm{s}^{(t+1)};\bm{\theta})$ and assess the action with the value network $\bm{q} (\bm{s}^{(t+1)},{\bm{a}}^{(t+1)};\bm{w})$.\\
5. Update the policy network $\bm{\mu} (\bm{s}^{(t)};\bm{\theta})$ and the value network $\bm{q} (\bm{s}^{(t)},{\bm{a}}^{(t)};\bm{w})$, as presented in Section \ref{sec:trainign}.\\
6. Repeat Step 1-5 until the reward $r^{(t)}$ converges to a tolerable error $\varepsilon$.
\caption{DRL-based BCD for Downlink WSRM Beamforming}
\label{alg:alg3}
\end{algorithm}

\section{Numerical Examples}\label{sec:expm}
In this section, we demonstrate the superior performance of our DRL-based BCD algorithm in solving the WSRM problems and compare it with several baselines. The data and code to generate the WSRM problem instances, implement the BCD method and the DRL-based BCD to replicate the numerical experiments can be accessed at \url{https://github.com/convexsoft/DRL-based-BCD}.

\subsection{Performance of DRL-based BCD for Downlink WSRM with Fixed Precoders}

To clarify the experimental mechanism of our DRL-based BCD and to facilitate replication, we first demonstrate its performance on small-scale downlink WSRM problems with fixed precoders. The performance on larger-scale problems is presented in the subsection on comparison with baselines.

\subsubsection{Samples Generation}
\textcolor{black}{
To train the DRL-based BCD model, we generated $1 \times 10^5$ WSRM problem instances with fixed precoders and $3$ users. The numbers of antennas at the base station and at each user are $3$ and $2$, respectively. To ensure a diverse set of samples, the values of the channel gains, noise power, and precoders were randomly set. The only variables to be optimized are the powers for the users. The optimal solutions for the downlink WSRM problems in these samples can be obtained either by brute-force search or efficiently solved using the algorithms described in \cite{TanTSP11}, which satisfy the condition of having a positive quasi-inverse.}

\subsubsection{Model Configuration}\label{sec:model_conf}
As presented in Section \ref{sec:RL-based BCD}, both the policy network and the value network are constructed using deep neural networks. The specific configurations of these neural networks are described in TABLE \ref{tab:tab1}.
\begin{table}[]
  \centering
  \renewcommand\tabcolsep{3.0pt}
    \begin{tabular}{ c cc }
    \toprule
    \textbf{Setting} & \textbf{Policy network} & \textbf{Value network} \\
    \midrule
        Input layer & ReLu, $19$   &ReLu, $19$   \\  [1.1ex]
        Hidden layers &ReLu, $128$  & ReLu, $128$ \\  [1.1ex]
        Output layer & Sigmoid, $3$   & Identified, $ 1$   \\  [1.1ex]
        Num of layers & $5$  & $5$  \\  [1.1ex]
    \bottomrule
    \end{tabular}%
    \captionsetup{font={footnotesize, stretch=1}, justification=raggedright}
    \caption{Summary of the important parameters of DRL-based BCD for the WSRM problem.}
  \label{tab:tab1}%
\end{table}%
We used two DRL techniques, DDPG and TD3, to construct the DDPG-based BCD and TD3-based BCD algorithms. 
We trained the policy and value network using mini-batch stochastic gradient descent. 
Some key parameters are set as in TABLE \ref{tab:tab2} for training the policy and value networks.
\begin{table}[htbp]
  \centering
  \renewcommand\tabcolsep{3.0pt}
        \begin{tabular}{ c cc }
    \toprule
    \textbf{Parameters} & \textbf{DDPG-based BCD} & \textbf{TD3-based BCD} \\
    \midrule
    Maximum step & $4500$ & $5000$   \\  [1.1ex]
    Learning rate & $1\times 10^{-6}$ & $1\times 10^{-6}$ \\  [1.1ex]
    Update frequency & $--$  & $200$ \\  [1.1ex]
    Noise & $U(0, 1\times 10^{-3})$ & $U(0, 1\times 10^{-3})$\\  [1.1ex]
    Replay buffer size & $5\times 10^{3}$  & $1\times 10^{6}$  \\  [1.1ex]
    Batch size & $20$  & $20$  \\  [1.1ex] 

    \bottomrule
    \end{tabular}%
    \captionsetup{font={footnotesize, stretch=1}, justification=raggedright}
    \caption{The summary of some key parameters of DRL-based BCD for the WSRM problem.}
  \label{tab:tab2}%
\end{table}%

\subsubsection{Algorithm Results}

\begin{figure*}[htbp]
\centering
\vspace{-8mm}
\subfigure[]{
\begin{minipage}[t]{0.25\linewidth}
\centering
\includegraphics[width=1.8in]{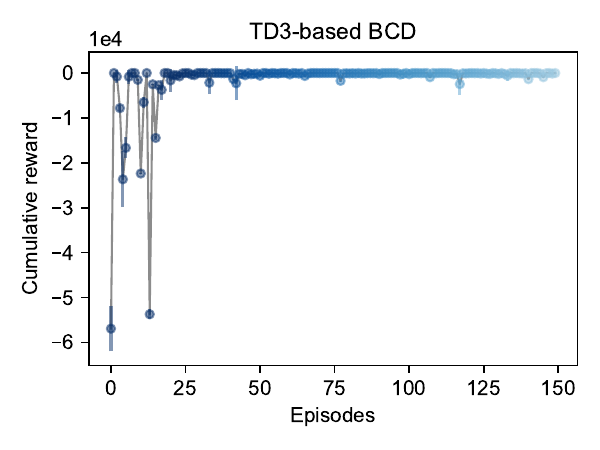}
\end{minipage}%
}%
\subfigure[]{
\begin{minipage}[t]{0.25\linewidth}
\centering
\includegraphics[width=1.8in]{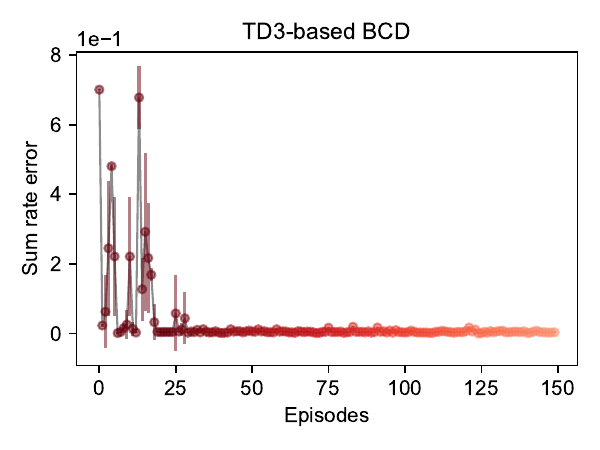}
\end{minipage}%
}%
\subfigure[]{
\begin{minipage}[t]{0.25\linewidth}
\centering
\includegraphics[width=1.8in]{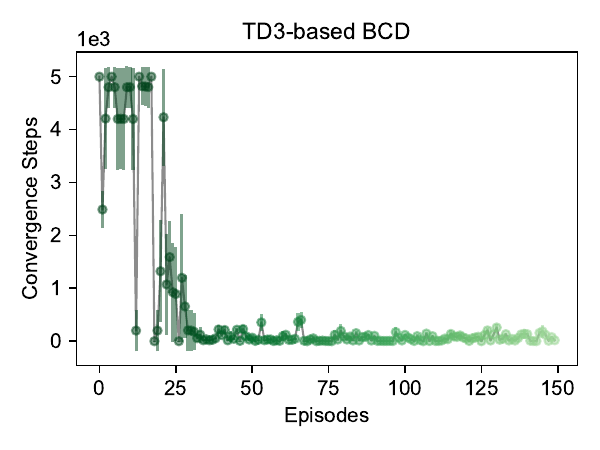}
\end{minipage}%
}%
\vspace{-3mm}
\subfigure[]{
\begin{minipage}[t]{0.25\linewidth}
\centering
\includegraphics[width=1.8in]{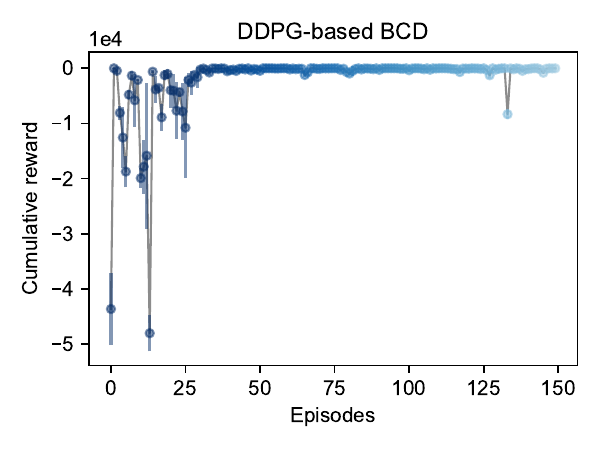}
\end{minipage}%
}%
\subfigure[]{
\begin{minipage}[t]{0.25\linewidth}
\centering
\includegraphics[width=1.8in]{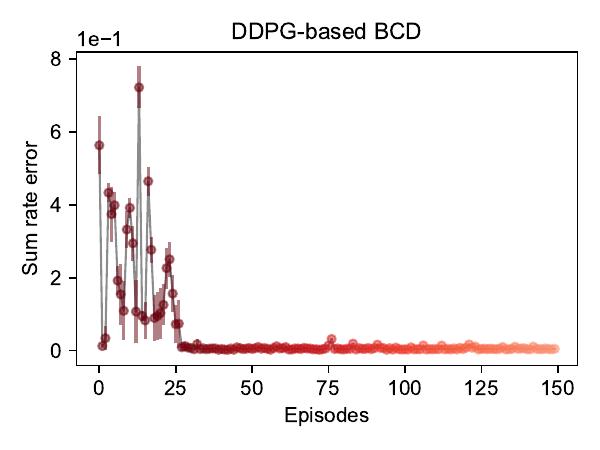}
\end{minipage}%
}%
\subfigure[]{
\begin{minipage}[t]{0.25\linewidth}
\centering
\includegraphics[width=1.8in]{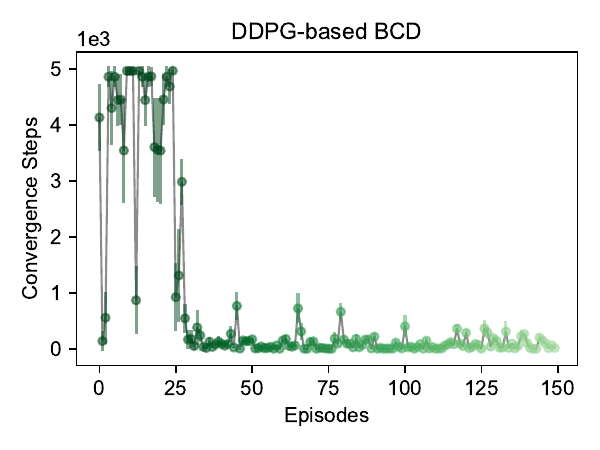}
\end{minipage}%
}%
\centering
\captionsetup{font={footnotesize}, justification=raggedright}
\caption{\textcolor{black}{Illustration with error bars on the performance of cumulative reward, sum rate error, and convergence steps for the TD3-based BCD model (a-c) and the DDPG-based model (d-f) in Algorithm \ref{alg:alg2} in solving WSRM problems.}}
\label{fig:rl-performance}
\end{figure*}

\begin{figure*}[htbp]
\centering
\vspace{-5mm}
\subfigure[]{
\begin{minipage}[t]{0.25\linewidth}
\centering
\includegraphics[width=1.8in]{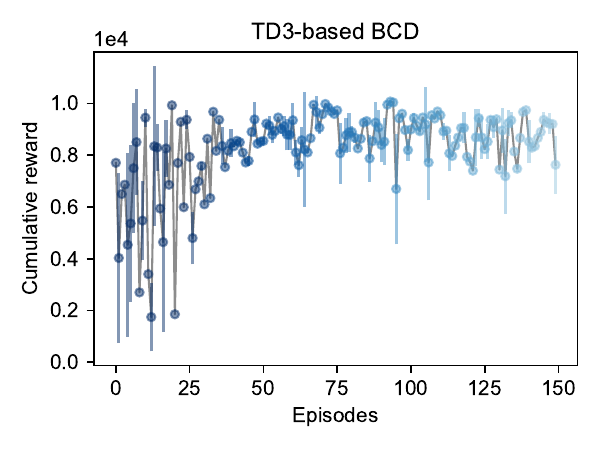}
\end{minipage}%
}%
\subfigure[]{
\begin{minipage}[t]{0.25\linewidth}
\centering
\includegraphics[width=1.8in]{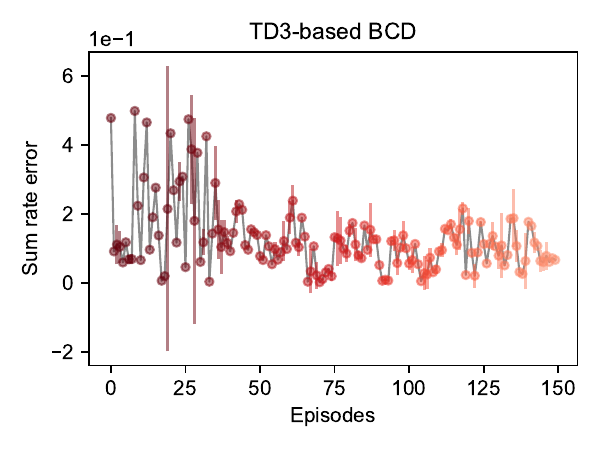}
\end{minipage}%
}%
\subfigure[]{
\begin{minipage}[t]{0.25\linewidth}
\centering
\includegraphics[width=1.8in]{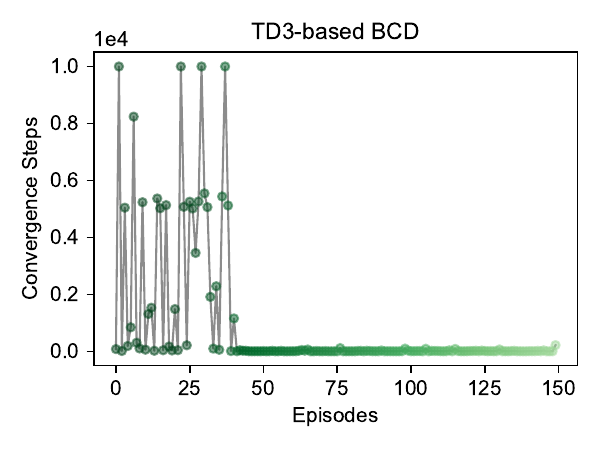}
\end{minipage}%
}%
\vspace{-3mm}

\subfigure[]{
\begin{minipage}[t]{0.25\linewidth}
\centering
\includegraphics[width=1.8in]{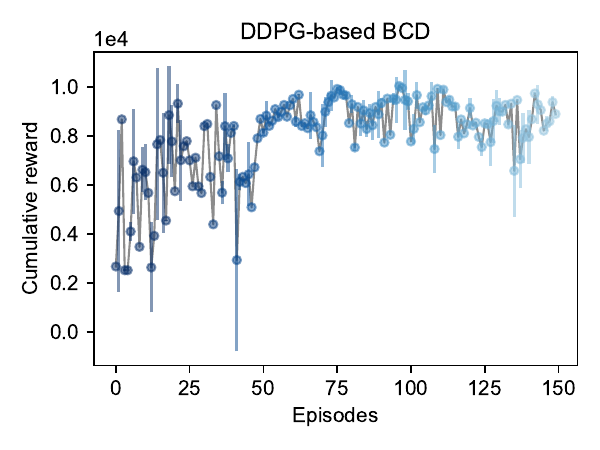}
\end{minipage}%
}%
\subfigure[]{
\begin{minipage}[t]{0.25\linewidth}
\centering
\includegraphics[width=1.8in]{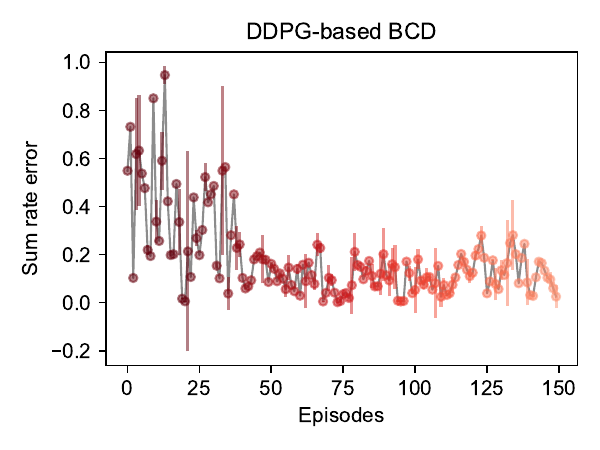}
\end{minipage}%
}%
\subfigure[]{
\begin{minipage}[t]{0.25\linewidth}
\centering
\includegraphics[width=1.8in]{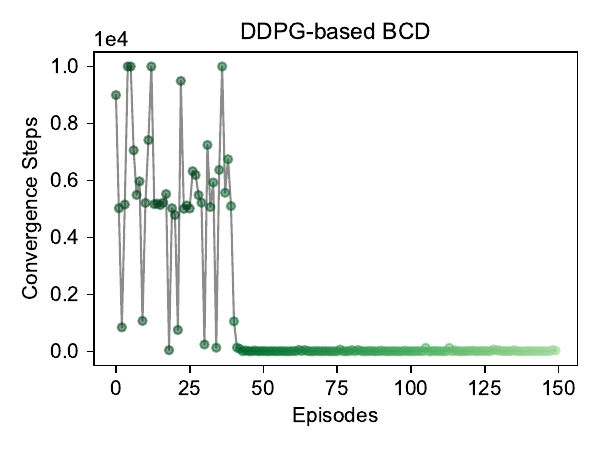}
\end{minipage}%
}%
\centering
\captionsetup{font={footnotesize}, justification=raggedright}
\caption{\textcolor{black}{Illustration with error bars on the performance of cumulative reward, sum rate error, and convergence steps for the TD3-based BCD model (a-c) and the DDPG-based model (d-f) using sum rate as the reward function in solving WSRM problems.}}
\label{fig:drl_bcd_sr_rewardf}
\end{figure*}


We evaluate the performance of the DRL-based BCD algorithm using three metrics: cumulative reward, sum rate error, and convergence steps. The cumulative reward reflects the algorithm's ability to improve the sum rate during each episode. The sum rate error measures the difference between the predicted and optimal sum rates. 
The Convergence steps indicate the number of iterations required to achieve convergence. Fig. \ref{fig:rl-performance} (a)-(c) and (d)-(e) present these metrics for the DDPG-based and TD3-based BCD, demonstrating that both models perform excellently.

\subsection{Comparison with DRL-based BCD Using Sum Rate as Reward Function}
We now compare the performance of the DRL-based BCD for downlink WSRM problems using fixed precoders with an alternative reward function instead of \eqref{eq:reward}. Here, the policy network, state transition, and value network remain consistent with the DRL-based BCD in Algorithm 2; however, the reward function is redefined as the sum rate. The long-term cumulative discounted reward, denoted as $\mathcal{V}^{(t)}$, is then expressed as
\begin{align}
    \mathcal{V}^{(t)}\!=\!\mathbb{E}\left\{\sum_{i=t}^{\infty} \alpha^i \sum_{l=1}^{L}w_l\log(1\!\!+\!\!\textup{SINR}_l\left(\mathbf{p}^{(i)}\right)) \!\mid\! s^{(t)}\right\},
\end{align}
where $\alpha$ represents the discount factor. 
This reward function has been explored in various studies, including \cite{khan2020centralized,jin2025sum}. By maximizing the cumulative discounted reward, the actor aims to achieve the maximum sum rate more quickly through self-exploration, thereby reducing the need for manually labeled optimal solutions and their associated costs. However, as shown in the results in Fig. \ref{fig:drl_bcd_sr_rewardf} (a)-(c) and (e)-(g), using the sum rate as the reward function leads to a cumulative reward that is significantly more unstable compared to the DRL-based BCD using reward as \eqref{eq:reward}. The sum rate error is also larger and more variable, with no notable improvement in convergence steps. Furthermore, using the sum rate as the reward function requires more training episodes for the DRL-based BCD models to achieve convergence.
Therefore, while the adoption of this reward function offers certain advantages, it also introduces notable challenges. In the absence of optimal solution labels for guidance, the model may engage in excessive exploration while relying on limited existing information, which can result in the model converging to local optima, thereby compromising the overall effectiveness of sum rate optimization. Furthermore, this approach generally requires extended training periods and a considerable volume of data and interactions to derive effective strategies.

\subsection{Comparison with Baselines}

\begin{figure*}[htbp]
\centering
\vspace{-8mm}
\subfigure[]{
\begin{minipage}[t]{0.215\linewidth}
\centering
\includegraphics[width=1.7in]{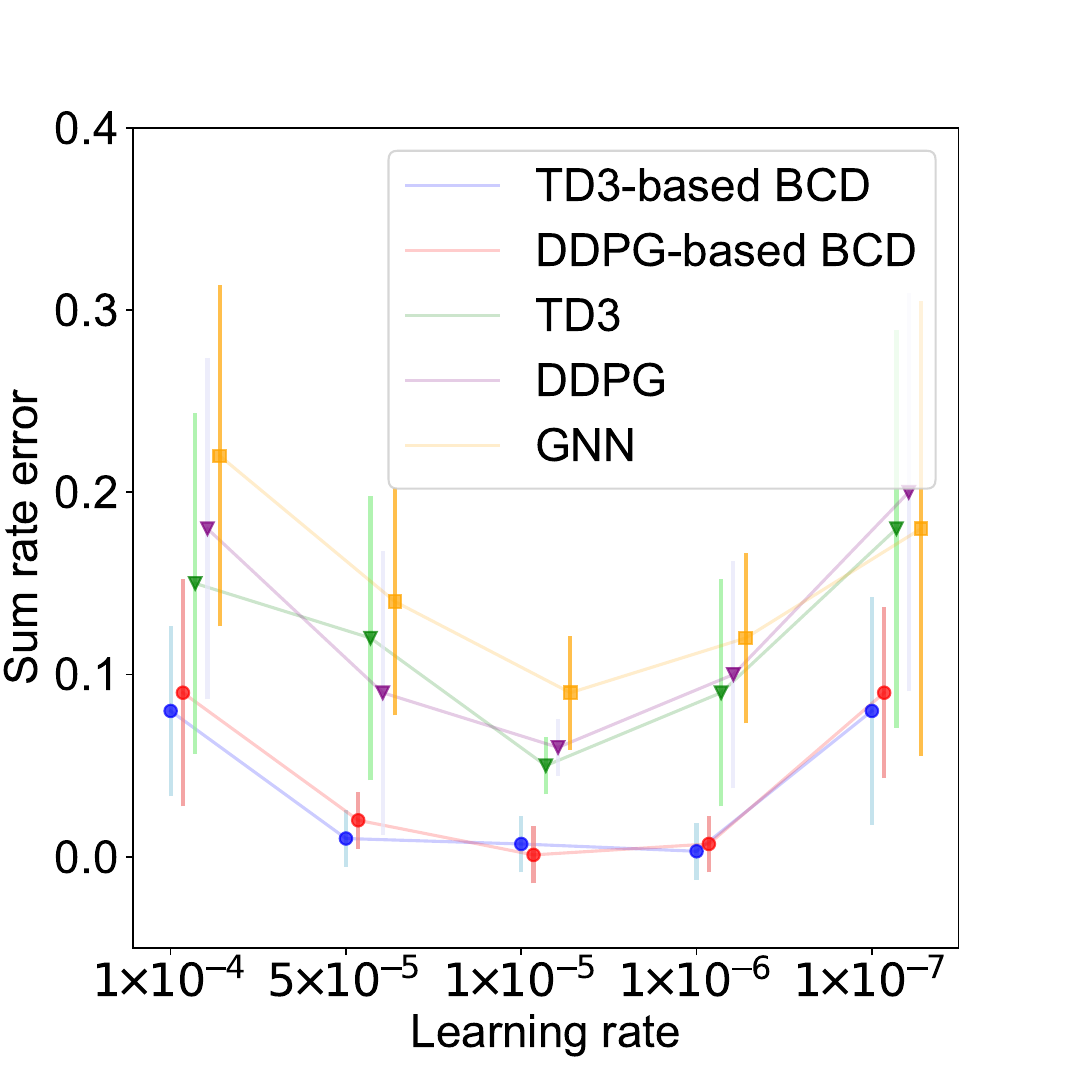}
\end{minipage}%
}%
\subfigure[]{
\begin{minipage}[t]{0.215\linewidth}
\centering
\includegraphics[width=1.7in]{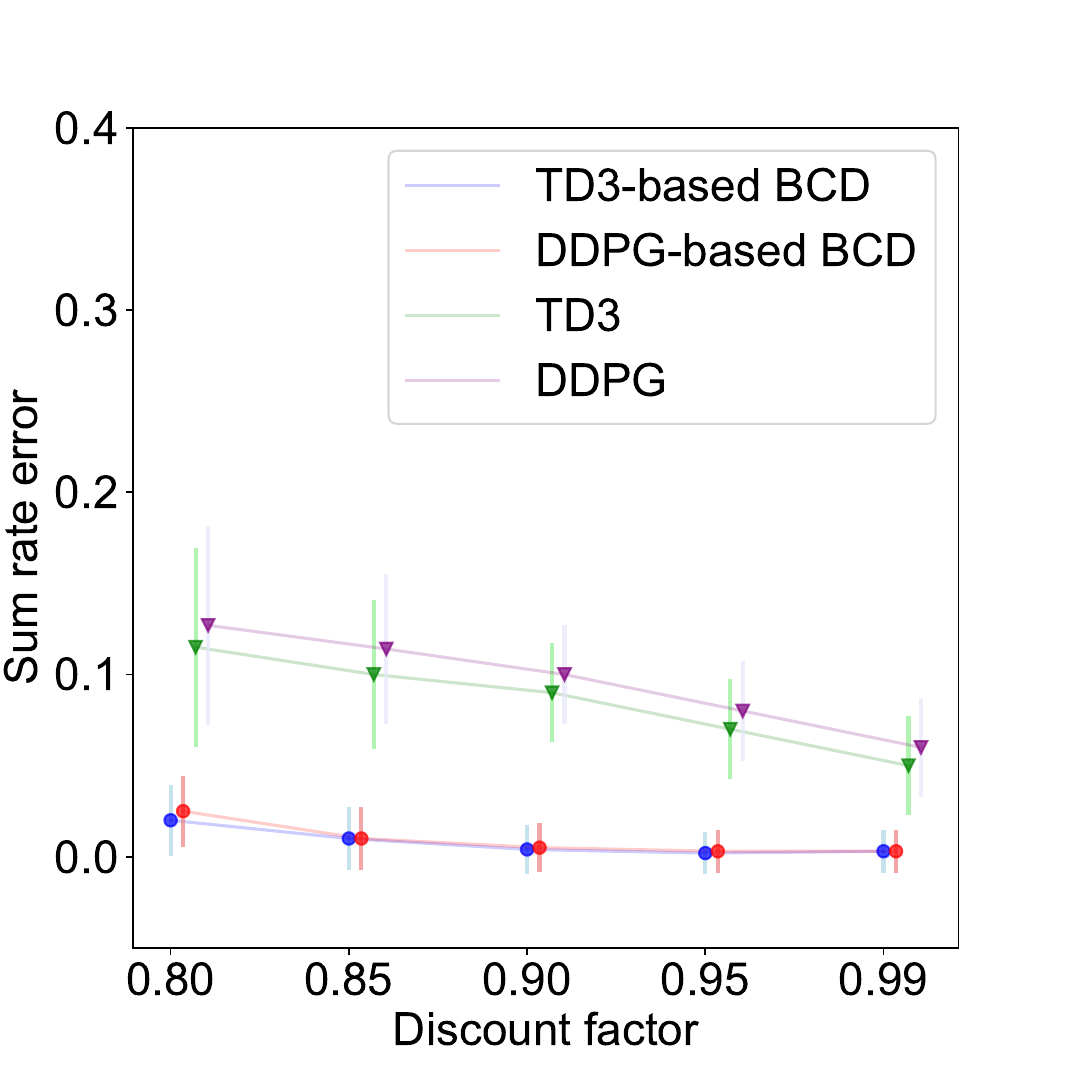}
\end{minipage}%
}%
\subfigure[]{
\begin{minipage}[t]{0.215\linewidth}
\centering
\includegraphics[width=1.7in]{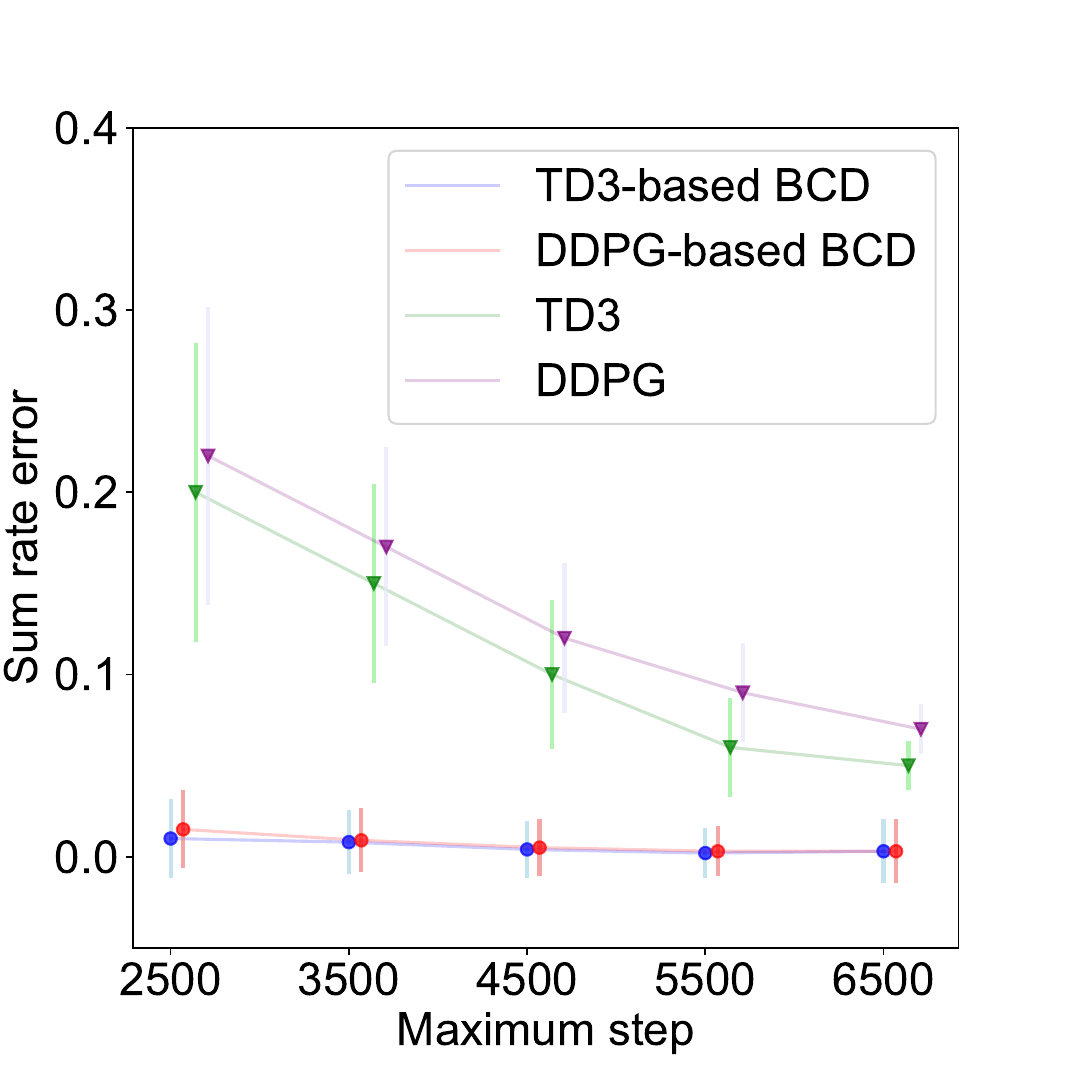}
\end{minipage}%
}%
\subfigure[]{
\begin{minipage}[t]{0.215\linewidth}
\centering
\includegraphics[width=1.7in]{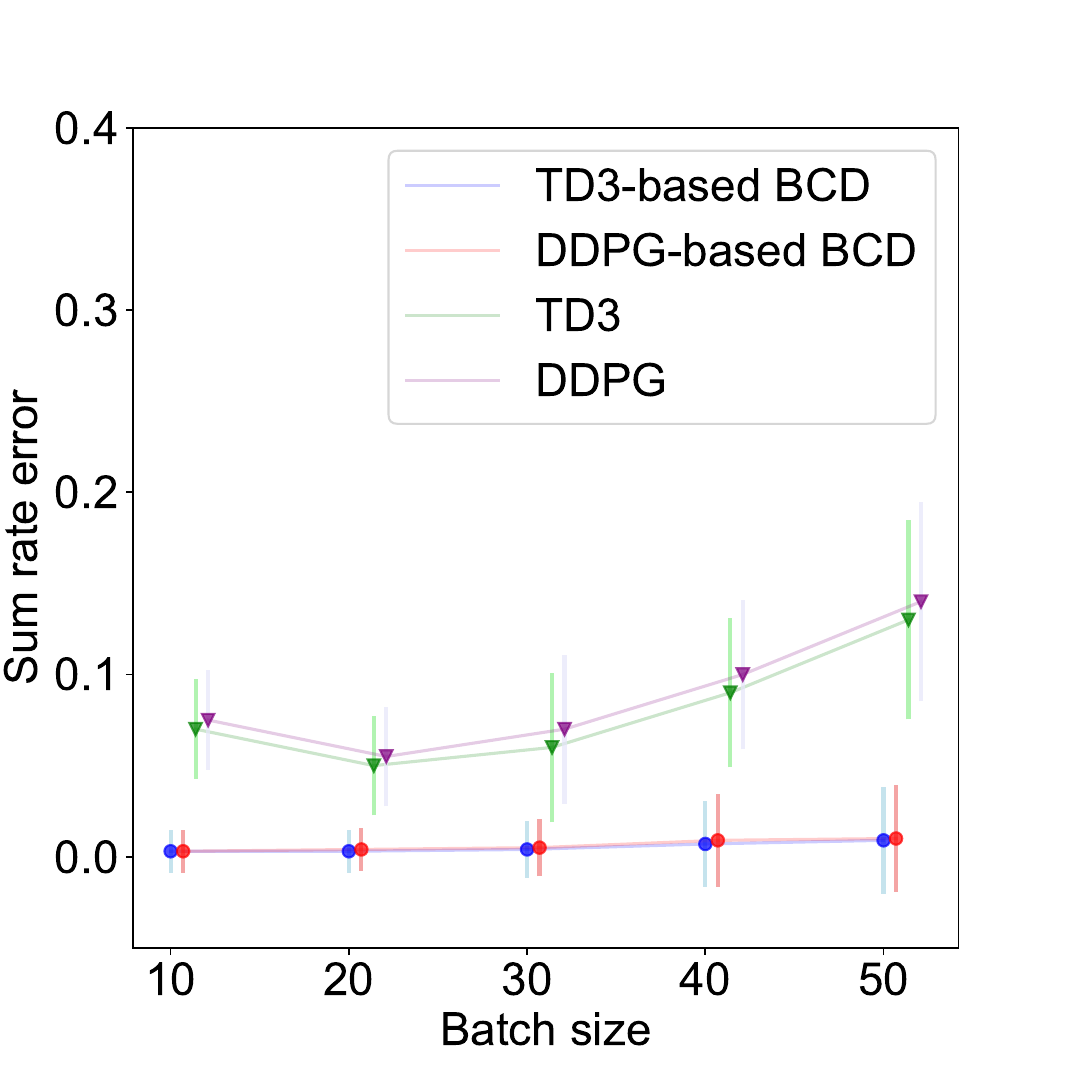}
\end{minipage}%
}%
\centering
\captionsetup{font={footnotesize}, justification=raggedright}
\caption{\textcolor{black}{ Illustration of the robustness of the proposed DRL-based BCD models compared to various purely deep learning baselines (particularly the purely DRL models) by examining the error bars of sum rate errors across different hyperparameter settings, such as learning rate (a), discount factor (b), maximum steps (c), and batch size (d).}}
\label{fig:robust_plot}
\end{figure*}

Now we compare the performance of the DRL-based BCD model with established baselines in solving larger-scale downlink WSRM problems with fixed precoders. The baselines we consider include the weighted minimum mean square error (WMMSE) \cite{shi2011iteratively}, successive convex approximation algorithm (SCA) \cite{chiang2007power}, the purely DDPG algorithm \cite{rahmani2022deep}, the purely TD3 algorithm \cite{waraiet2023robust} and the GNN algorithm \cite{lee2020graph}. The designs of the deep learning-based baselines are briefly described below:

\begin{figure*}[htbp]
\centering
\vspace{-4mm}
\subfigure[]{
\begin{minipage}[t]{0.215\linewidth}
\centering
\includegraphics[width=1.7in]{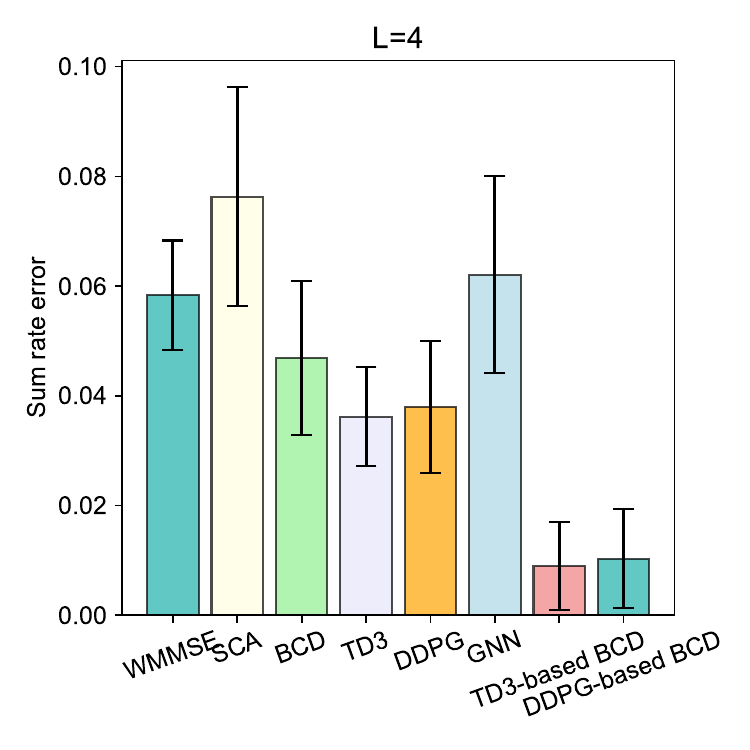}
\end{minipage}%
}%
\subfigure[]{
\begin{minipage}[t]{0.215\linewidth}
\centering
\includegraphics[width=1.7in]{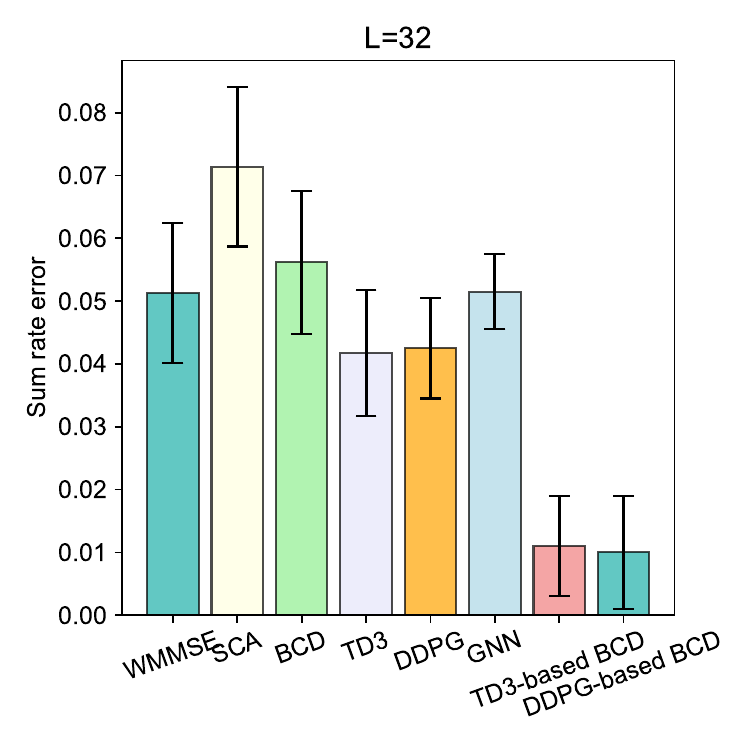}
\end{minipage}%
}%
\subfigure[]{
\begin{minipage}[t]{0.215\linewidth}
\centering
\includegraphics[width=1.7in]{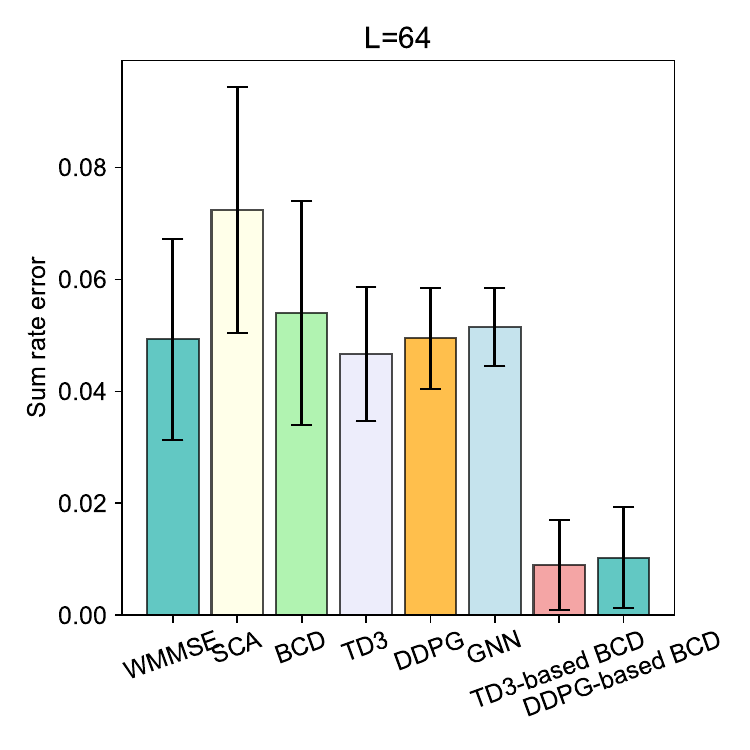}
\end{minipage}%
}%
\subfigure[]{
\begin{minipage}[t]{0.215\linewidth}
\centering
\includegraphics[width=1.7in]{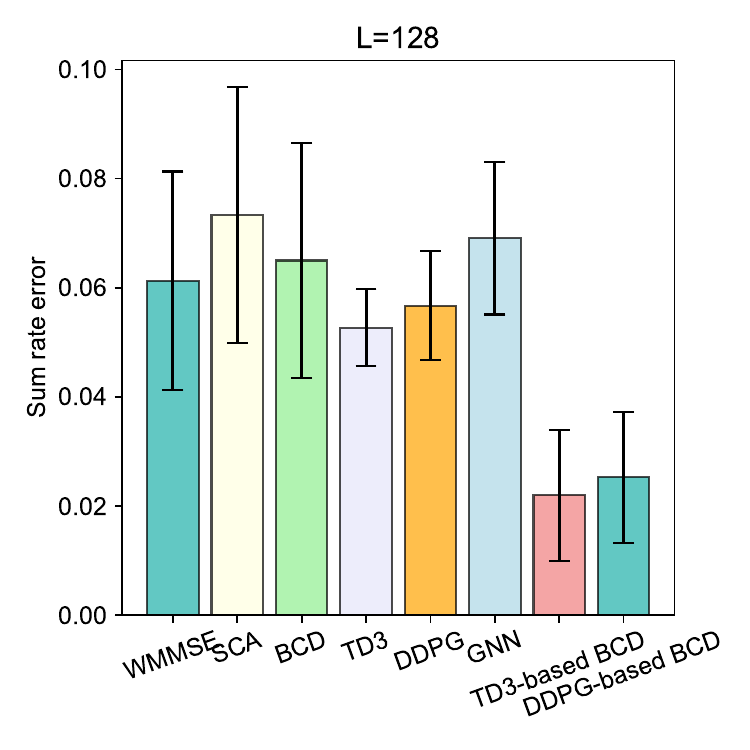}
\end{minipage}%
}%
\centering
\vspace{-2mm}
\captionsetup{font={footnotesize}, justification=raggedright}
\caption{\textcolor{black}{Illustration of the effectiveness of the proposed DRL-based BCD models compared to various baselines by examining the error bars for sum rate errors under different network scales.}}
\label{fig:sr_plot}
\end{figure*}

\begin{itemize}
    \item DDPG \cite{rahmani2022deep}: The structure of DDPG consists of an actor DNN, denoted as $\bm{\mu}\left(\bm{s}^{(t)}, \theta_\mu\right)$, and a critic DNN, denoted as $q\left(\bm{s}^{(t)}, \bm{a}^{(t)}, \theta_q\right)$. We refer to the approach presented in \cite{rahmani2022deep}, which utilizes DDPG to address the WSRM problem, a purely deep learning model. In other words, apart from the general parameters for the WSRM problem, minimal additional analytical model knowledge has been incorporated into it. The environment state $\bm{s}^{(t)}$ consists of the current user's SINRs, current power, the current interference, and objective gradients evaluated at the current power. Considering $L$ users in the wireless system, the dimension of the environment state is $3L$. Specifically, the state $\bm{s}^{(t)}$ at the $t$-th step is defined as
\begin{align*}
\bm{s}^{(t)}= \Big\{\bm{\gamma^{(t)}}, \p^{(t)},\mathbf{F}\p^{(t)}+\bm{\sigma},\textup{diag}(\mathbf{w})(1/(1+\bm{\gamma}^{(t)}))\Big\}. 
\end{align*}
The action is the change in the transmit power of the users, i.e., $\bm{a}^{(t)}=\Delta \p^{(t)} \in \mathbb{R}^L$, where the transmitted power of the user at the $t$-th step is
$$
\p^{(t+1)}=\p^{(t)}+\Delta \p^{(t)} .
$$
Finally, we choose the reward function as the objective function, i.e., the weighted sum rate.
\item TD3 \cite{waraiet2023robust}: Similarly, this baseline is also a purely deep learning model without additional analytical knowledge added to it. To train the policy network and the value network, we referred to the TD3 technique and another set of events for learning by referring to \cite{waraiet2023robust}. The state $\bm{s}(t)$ for this baseline is designed as
$$
\begin{aligned}
\bm{s}^{(t)}= \left\{\G, \p^{(t-1)},\textup{diag}(\w)\log(1+\bm{\gamma}^{(t)})\right\}.
\end{aligned}
$$
Besides, action $\bm{a}^{(t)}$ and reward $r^{(t)}$ are respectively designed as the transmit power $\bm{a}^{(t)}=\p^{(t)}$ and the weighted sum rate $\sum_{l=1}^{L}w_l\log(1+\gamma_l^{(t)})$.
\item GNN \cite{lee2020graph}: GNNs have also been a popular model used to address the WSRM problem. We refer to the work \cite{lee2020graph} to leverage GNN for the downlink sum-rate maximization problem. We represent the communication system as a graph $\bm{\mathcal{G}}(\bm{\mathcal{V}}, \bm{\mathcal{E}})$, where node set $\bm{\mathcal{V}}$ represents users and edge set $\bm{\mathcal{E}}$ represents interference between users. Node features include weights, noise power $n_i$, and maximum allowable power, while edge features represent channel gain $G_{i j}$. GNN utilizes the graph structure to capture user dependencies and interactions, updating node states by aggregating neighboring information. The GNN outputs optimal transmit power to maximize the sum-rate. Training involves optimization using techniques like backpropagation and gradient descent.
\end{itemize}




\subsubsection{Comparison of Robustness}
As is well known, DRL models typically rely on various hyperparameters, including learning rate, discount factor, and exploration strategy. The selection of these hyperparameters has a significant impact on model performance, often necessitating extensive experimentation to identify the optimal combination.
\textcolor{black}{Here, we demonstrate the robustness of our proposed DRL-based BCD model by comparing its performance stability, measured in terms of sum rate error, against various baseline models under different hyperparameter configurations.}
The baselines we consider include purely DRL models, specifically the TD3 and DDPG models, as well as a GNN model. \textcolor{black}{The results are presented for scenarios involving three users, where the transmitter and receiver are equipped with three and two antennas, respectively. We analyze the impact of four hyperparameters on model performance: learning rate, discount factor, maximum steps, and batch size.}
The comparative results are illustrated in Fig. \ref{fig:robust_plot}.
The results show that our DRL-based BCD models generally achieve lower rate errors than other baseline models. While various learning rate settings similarly affect both our model and the baselines, our model consistently performs better. The discount factor (ranging from 0.8 to 0.99) has a negligible impact on the performance of the DRL-based BCD models, and the ranges for maximum steps (from 2500 to 6500) and batch size (from 10 to 50) also have minimal effects on the sum rate error.
Conversely, varying the discount factor, maximum steps, and batch size significantly affects the purely DRL models, with maximum steps having the greatest impact. This suggests that the DRL model requires more steps to explore effectively without analyzing the intrinsic structure of the WSRM problem. Overall, these findings demonstrate that our DRL-based BCD models provide greater robustness and improved performance.

\subsubsection{Comparison of Effectiveness}
\textcolor{black}{
To compare the effectiveness of the DRL-based BCD with other baselines, we investigate systems of four different sizes: \(L=4\), \(L=32\), \(L=64\), and \(L=128\). The number of transmitter antennas corresponds to each network size, set as \(N=4\), \(N=32\), \(N=64\), and \(N=128\), respectively. The number of receiver antennas is fixed at \(N_l=2\). }
To create the training set for all models, we generate problem instances based on the conditions in the FA algorithm from \cite{TanTSP11}, which efficiently yields optimal solutions. Additionally, we use brute-force search methods to obtain further optimal solutions for scenarios where $L=4$. In practice, training samples can also be derived from historical data. Fig. \ref{fig:sr_plot} shows the sum rate error for our proposed BCD algorithm, the DRL-based BCD algorithm, and the baselines. The DDPG-based and TD3-based BCD algorithms demonstrate the highest accuracy in weighted sum rate across various network scales, with the smallest variance, indicating greater stability in the DRL-based BCD models.

\subsubsection{Comparison of Convergence Speed}

Here, we compare the convergence speed of our DRL-based BCD models with pure DDPG, pure TD3, and BCD. Convergence speed is measured by the number of steps the actor network takes to converge. As shown in Fig. \ref{fig:convergence}, although BCD converges fastest, it exhibits the highest average and variance in sum rate error, indicating poorer performance. The DRL-based BCD actor networks converge faster than pure DRL models, but their environment updates require additional iterations, potentially increasing computational time. Nevertheless, the DRL-based BCD model achieves higher precision and sum rate, which may offset this cost. Additionally, it is more robust and requires less hyperparameter tuning, facilitating easier convergence and reducing training costs.

\begin{figure}
\centerline{\includegraphics[scale=0.28]{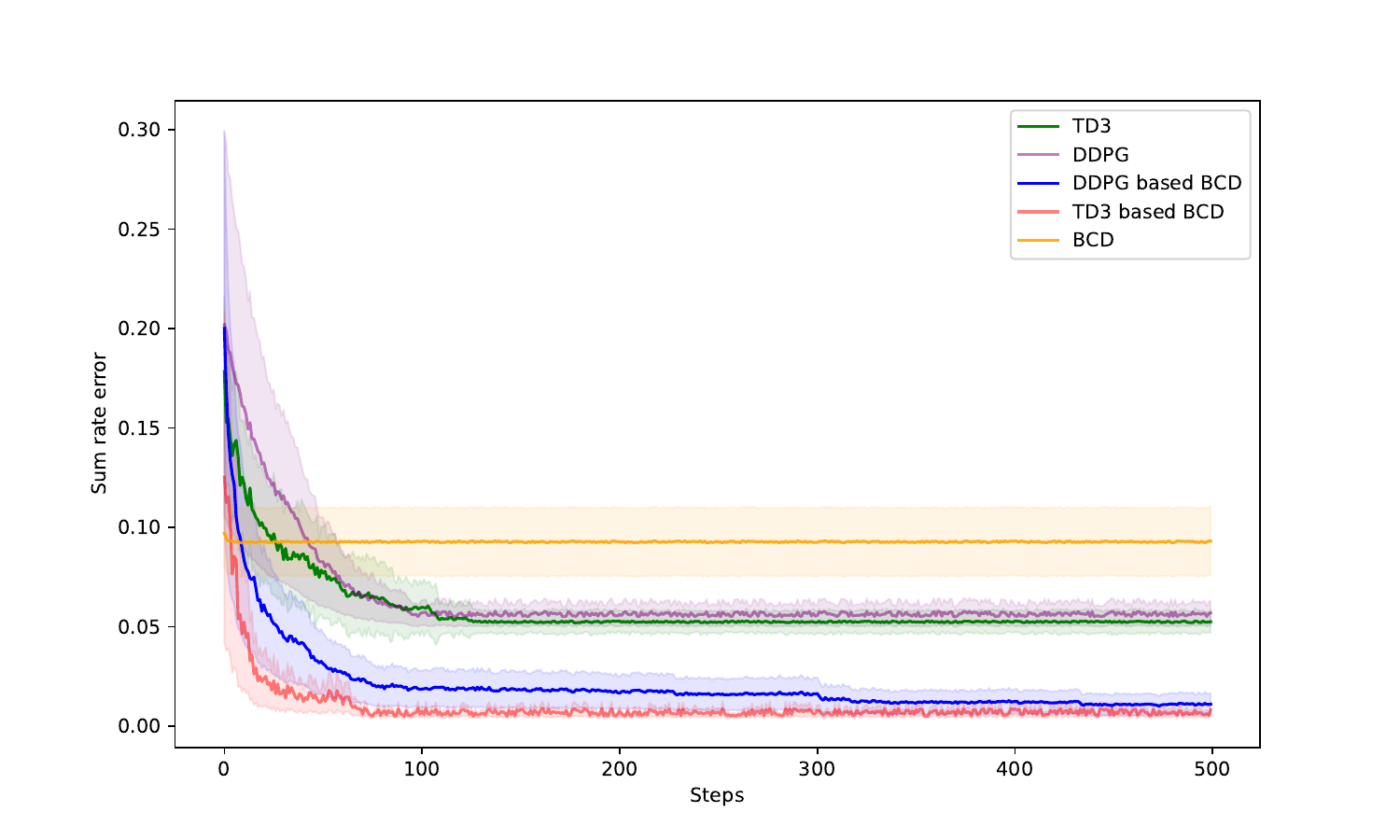}}
\captionsetup{font={footnotesize, stretch=1}, justification=raggedright}
\caption{\textcolor{black}{Illustration of the convergence speed of the DRL-based BCD models compared to the purely BCD and purely DRL models by examining the evolution of sum rate errors as the number of action steps increases. }
\label{fig:convergence}}\vspace{-2mm}
\end{figure}

\subsubsection{Comparison of Efficiency}

\textcolor{black}{
We evaluate the efficiency of the proposed DRL-based BCD model in comparison to several learning-based models as baselines, specifically DDPG, TD3, and GNN. Model efficiency is assessed using training time ($T_{\textup{Training}}$) and testing time ($T_{\textup{Testing}}$). The results are presented in TABLE \ref{tab:time}. 
In general, DRL-based models often do not provide significant advantages in either training or testing time compared to the GNN models, as also demonstrated by our results. However, it is important to note that our model's performance, particularly in terms of sum rate accuracy, significantly surpasses that of the GNN model. On the other hand, although our DRL-based BCD models do not exhibit a substantial advantage in testing time compared to the purely DRL models, they require less training time. This is primarily because, during testing, while the actor networks of the DRL-based BCD models converge in fewer steps, their environment updates necessitate additional iterations, potentially increasing computational time. However, during training, the DRL-based BCD model is more robust and requires less hyperparameter tuning, which facilitates convergence and reduces training costs. Furthermore, its sum rate accuracy is also superior to that of the purely DRL models.
}

\begin{table*}[h]
\footnotesize
\centering
\begin{tabular}{|c|cc|cc|cc|}
\hline
System Scales                                  & \multicolumn{2}{c|}{L=4}                                              & \multicolumn{2}{c|}{L=64}                                             & \multicolumn{2}{c|}{L=128}                                                                                      \\ \hline
Metrics                                        & \multicolumn{1}{c|}{$T_{\textup{Training}}(s)$} & $T_{\textup{Testing}}(s)$ & \multicolumn{1}{c|}{$T_{\textup{Training}}(s)$} & $T_{\textup{Testing}}(s)$ & \multicolumn{1}{c|}{$T_{\textup{Training}}(s)$} & $T_{\textup{Testing}}(s)$  \\ \hline
DDPG-based BCD                                 & \multicolumn{1}{c|}{$1.2\!\!\times\!\!10^{4}$    }                  & $0.022$             & \multicolumn{1}{c|}{$2.3\!\!\times\!\!10^{5}$}                  & $0.037$                 & \multicolumn{1}{c|}{$3.8\!\!\times\!\!10^{5}$}                  &$0.069$                     \\ \hline
TD3-based BCD                                  & \multicolumn{1}{c|}{$1.4\!\!\times\!\!10^{4}$}                  &$0.021$            & \multicolumn{1}{c|}{$2.2\!\!\times\!\!10^{5}$}                  &$0.035$                 & \multicolumn{1}{c|}{$3.6\!\!\times\!\!10^{5}$}                  & $0.065$                      \\ \hline
DDPG    \cite{rahmani2022deep}                                       & \multicolumn{1}{c|}{$1.9\!\!\times\!\!10^{4}$}                  & $0.009$               & \multicolumn{1}{c|}{$2.8\!\!\times\!\!10^{5}$}                  & $0.021$                 & \multicolumn{1}{c|}{$4.6\!\!\times\!\!10^{5}$}                  & $0.052$                       \\ \hline
TD3   \cite{waraiet2023robust}                                         & \multicolumn{1}{c|}{$1.9\!\!\times\!\!10^{4}$}                  & $0.009$               & \multicolumn{1}{c|}{$2.8\!\!\times\!\!10^{5}$}                  & $0.028$                 & \multicolumn{1}{c|}{$4.7\!\!\times\!\!10^{5}$}                  & $0.054$                     \\ \hline
GNN  \cite{lee2020graph}                                          & \multicolumn{1}{c|}{$3.9\!\!\times\!\!10^{3}$}                  & $0.007$                 & \multicolumn{1}{c|}{$1.3\!\!\times\!\!10^{5}$}                  & $0.017$                 & \multicolumn{1}{c|}{$2.1\!\!\times\!\!10^{5}$}                  & $0.045$                     \\ \hline
\end{tabular}
\centering
\captionsetup{font={footnotesize, stretch=1}, justification=raggedright}
\caption{Efficiency comparison with different scales of users for the downlink sum-rate maximization.}
\label{tab:time}
\end{table*}

\vspace{-5mm}
\subsection{Performance of DRL-based BCD for Downlink WSRM Beamforming}
Now we present the effective performance of the DRL-based BCD for downlink WSRM beamforming. 
\subsubsection{Algorithm Results}
\textcolor{black}{
We first evaluate the DRL-based BCD models for downlink WSRM beamforming in small-scale networks with $L=4$, $N=4$, and $N_l=2$. Here, we present only the performance of the TD3-based BCD. The network structure and hyperparameters follow those in Subsection \ref{sec:model_conf}, except for the environment state transition, which is defined in \eqref{eq:state2}. We utilize cumulative reward, sum rate error, and convergence steps as metrics for performance evaluation. The results in Fig. \ref{fig:drl_bcd_sr_rewardf} also show the strong performance of DRL-based BCD for downlink WSRM beamforming.
}
\textcolor{black}{
\subsubsection{Beamforming Comparison}
Now we compare the BCD algorithm in Algorithm \ref{alg:alg1}, which utilizes a fixed suboptimal linear precoder for designing transmitter and receiver beamforming vectors, with a variant that employs interchangeably optimized precoders, as described in Section \ref{sec:extend} (denoted as IP-BCD). The suboptimal linear precoders for the transmitter and receiver beamforming are defined as $\mathbb{U} = \mathbf{H}^{\dagger} \mathbf{H} (\mathbf{H}^{\dagger} \mathbf{H} \mathbf{H})^{-1} \mathbf{D}$ and $\mathbf{v}_l = \frac{\mathbf{H}_l \mathbf{u}_l}{\|\mathbf{H}_l \mathbf{u}_l\|}$, where $\mathbf{D}$ is a diagonal scaling matrix ensuring $\|\mathbf{u}_l\|=1$.
We also compare the DRL-based BCD in Algorithm \ref{alg:alg2}, which uses the suboptimal linear precoder, with the DRL-based BCD in Algorithm 3, which employs interchangeably optimized precoders (denoted as IP-DRL-based BCD). We consider the following cases: $L=4, N=4, N_l=2$; $L=4, N=64, N_l=2$; $L=32, N=32, N_l=2$; and $L=32, N=128, N_l=2$. The training datasets are generated according to \cite{lai2014beamforming}. The Performance is evaluated based on sum rate accuracy and computational time (for the trained models), and the results are shown in Fig. \ref{fig:LP_IP_sr_plot} and TABLE \ref{tab:LP_IP_time}, respectively.
It is observed that for both $L=4$ and $L=32$, IP-BCD outperforms BCD, and IP-DRL-based BCD outperforms DRL-based BCD. However, as the number of transmitter antennas increases, the performance advantage of IP-BCD and IP-DRL-based BCD decreases. Additionally, the computational time of IP-BCD and IP-DRL-based BCD is higher than that of BCD and DRL-based BCD. Therefore, IP-BCD and IP-DRL-based BCD may be more suitable for scenarios with fewer transmitter antennas or where more accurate solutions are required. For larger numbers of transmitter antennas, using BCD in Algorithm \ref{alg:alg1} and DRL-based BCD Algorithm \ref{alg:alg2} with suboptimal linear precoders may be sufficient or more practical.  
}
\vspace{-0mm}
\begin{figure*}[htbp]
\vspace{-3mm}
\centering
\subfigure[]{
\begin{minipage}[t]{0.25\linewidth}
\centering
\includegraphics[width=1.8in]{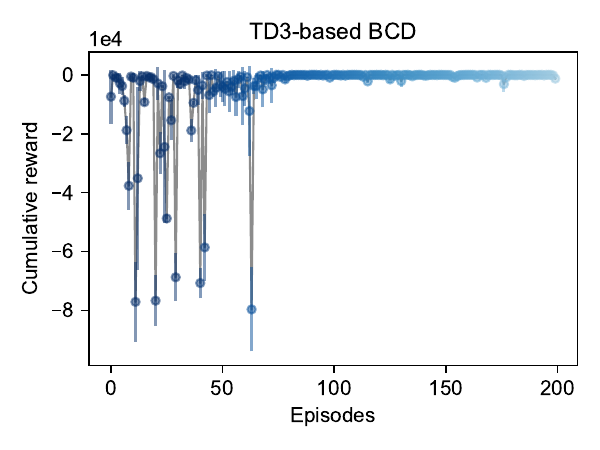}
\end{minipage}%
}%
\subfigure[]{
\begin{minipage}[t]{0.25\linewidth}
\centering
\includegraphics[width=1.8in]{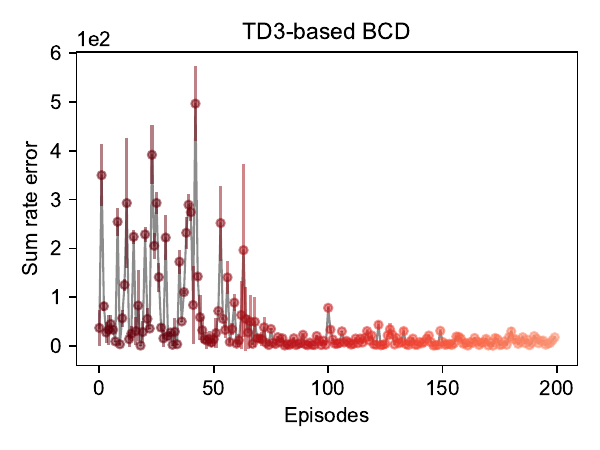}
\end{minipage}%
}%
\subfigure[]{
\begin{minipage}[t]{0.25\linewidth}
\centering
\includegraphics[width=1.8in]{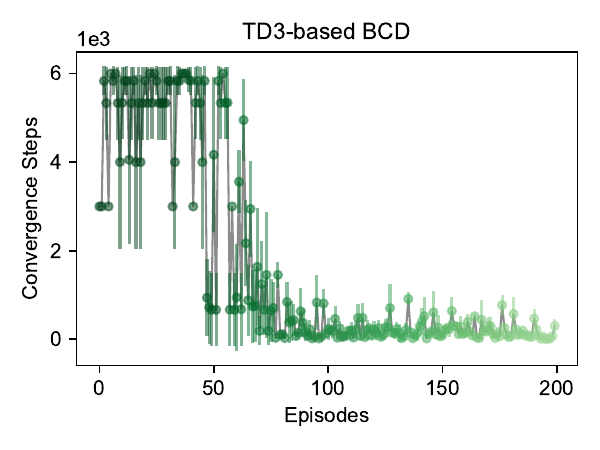}
\end{minipage}%
}%
\vspace{-1mm}
\centering
\captionsetup{font={footnotesize}, justification=raggedright}
\caption{\textcolor{black}{Illustration with error bars on the performance of cumulative reward (a), sum rate error (b), and convergence steps (c) for the TD3-based BCD model in Algorithm \ref{alg:alg3} in solving downlink WSRM beamforming.}}
\label{fig:drl_bcd_sr_rewardf}
\end{figure*}

\begin{figure*}[htbp]
\centering
\vspace{-4mm}
\subfigure[]{
\begin{minipage}[t]{0.215\linewidth}
\centering
\includegraphics[width=1.7in]{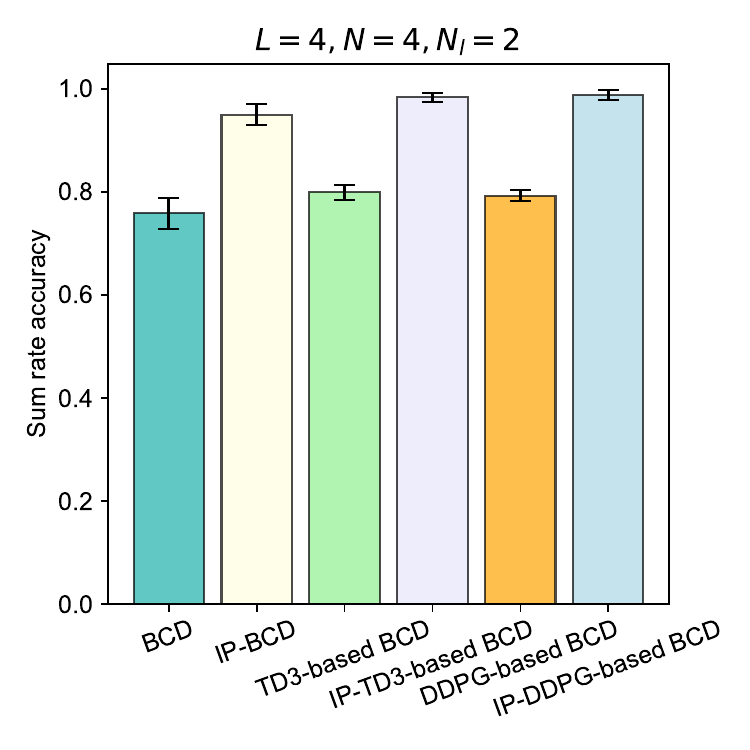}
\end{minipage}%
}%
\subfigure[]{
\begin{minipage}[t]{0.215\linewidth}
\centering
\includegraphics[width=1.7in]{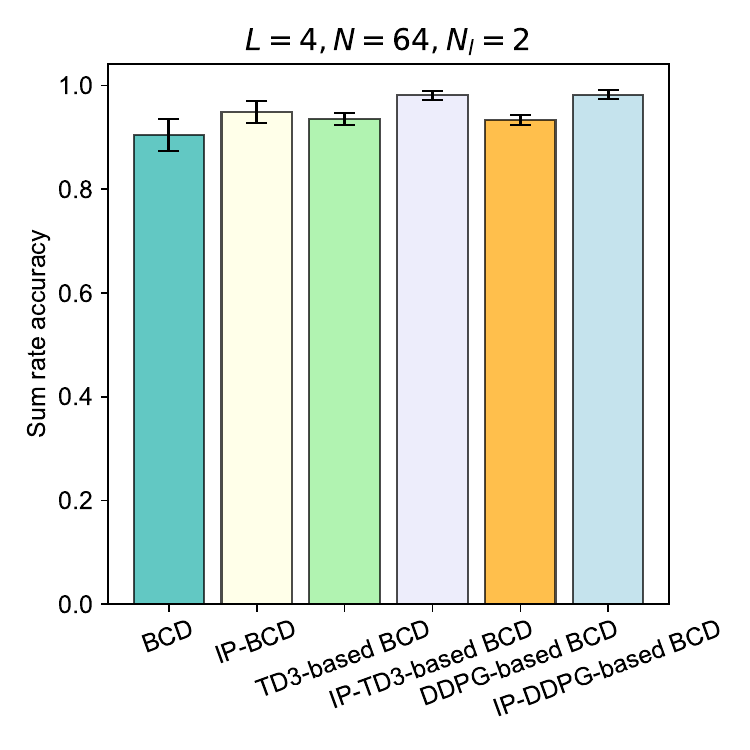}
\end{minipage}%
}%
\subfigure[]{
\begin{minipage}[t]{0.215\linewidth}
\centering
\includegraphics[width=1.7in]{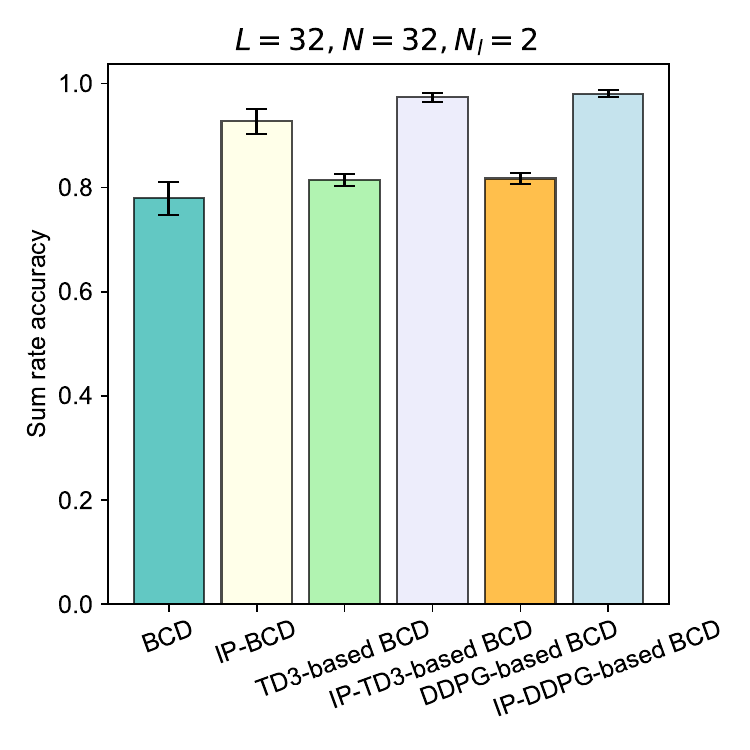}
\end{minipage}%
}%
\subfigure[]{
\begin{minipage}[t]{0.215\linewidth}
\centering
\includegraphics[width=1.7in]{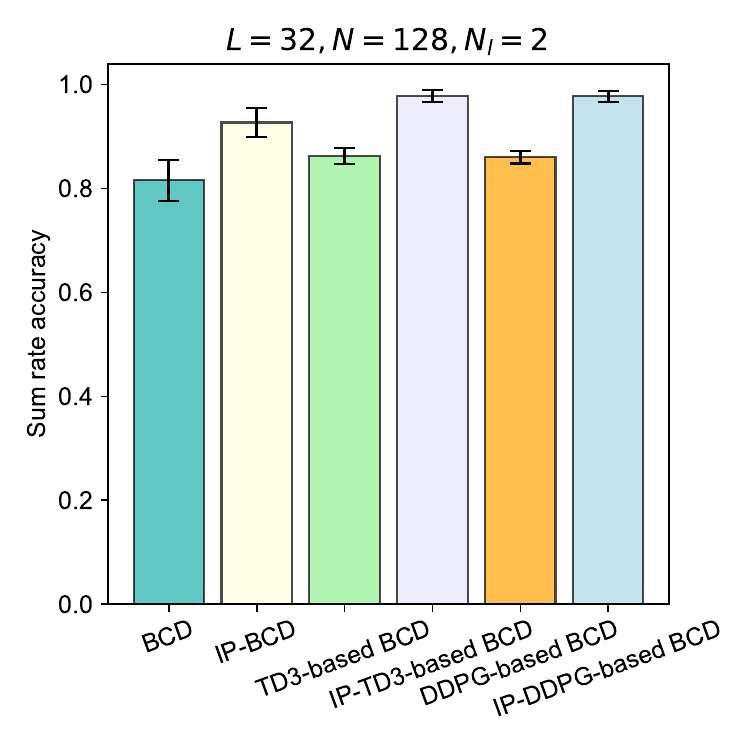}
\end{minipage}%
}%
\centering
\vspace{-2mm}
\captionsetup{font={footnotesize}, justification=raggedright}
\caption{\textcolor{black}{Illustration of the sum rate accuracy for the proposed BCD models and DRL-based BCD models, in comparison to IP-BCD model and IP-DRL-based BCD models across various network scales.}}
\label{fig:LP_IP_sr_plot}
\end{figure*}
\vspace{-3mm}

\begin{table*}[h]
\footnotesize
\centering
\begin{tabular}{|c|c|c|c|c|}
\hline
\multicolumn{1}{|c|}{Metric} & \multicolumn{4}{c|}{Computational time (s)}                                                       \\ \hline
System Scales           & $L\!\!=\!\!4,M\!\!=\!\!4,N_L\!\!=\!\!2$  & $L\!\!=\!\!4,M\!\!=\!\!64,N_L\!\!=\!\!2$  & $L\!\!=\!\!32,M\!\!=\!\!32,N_L\!\!=\!\!2$  & $L\!\!=\!\!32,M\!\!=\!\!128,N_L\!\!=\!\!2$  \\ \hline
BCD               & $5.9\times 10^{-3}$ & $5.7\times 10^{-3}$ & $7.4\times 10^{-3}$ & $8.0\times 10^{-3}$ \\ \hline
IP-BCD            & $1.8\times 10^{-1}$ & $3.4\times 10^{-1}$ & $1.4$ & $1.1\times 10^1$ \\ \hline
DDPG-based BCD    & $2.2\times 10^{-2}$ & $2.3\times 10^{-2}$ & $3.0\times 10^{-2}$ & $3.3\times 10^{-2}$ \\ \hline
IP-DDPG-based BCD & $4.5\times 10^{-1}$ & $9.6\times 10^{-1}$ & $1.9$ & $1.5\times 10^1$ \\ \hline
TD3-based BCD     & $2.1\times 10^{-2}$ & $2.1\times 10^{-2}$ & $2.9\times 10^{-2}$ & $3.0\times 10^{-2}$ \\ \hline
IP-TD3-based BCD  & $3.9\times 10^{-1}$ & $9.2\times 10^{-1}$ & $1.7$ & $1.3\times 10^1$  \\ \hline
\end{tabular}
\captionsetup{font={footnotesize, stretch=1}, justification=raggedright}
\caption{\textcolor{black}{Computational time comparison for the proposed BCD models and DRL-based BCD models, in comparison to IP-BCD model and IP-DRL-based BCD models across various network scales.}}
\label{tab:LP_IP_time}
\end{table*}



\section{Conclusion}\label{sec:conclusion}

In this study, we have introduced a DRL-based BCD model to solve the nonconvex downlink sum-rate maximization problem with a power constraint, combining deep learning and optimization theory to improve performance by reducing sensitivity to initial points and avoiding local optima. The model excelled in adhering to constraints, enhancing accuracy, and potentially achieving optimal solutions. Unlike traditional deep learning methods, it maintained interpretability and was easier to train by leveraging the problem's structure through theoretical analysis. Furthermore, we demonstrated that the DRL-based BCD framework is highly extensible and can be effectively applied to other scenarios, such as joint beamforming for sum rate maximization. Numerical experiments confirmed the algorithm's effectiveness, efficiency, and robustness. Future work will focus on integrating the mathematical structure of sum rate maximization with data-driven techniques to develop distributed algorithms with low computational complexity and exploring new optimization methods for other significant nonconvex problems in next-generation AI-native wireless networks.

\section{Appendix}

\subsection{Proof of Lemma \ref{thm:thm1}}
    Firstly, $\rho(\textup{diag}(e^{\Tilde{\bm{\gamma}}})\B)$ in \eqref{eq:sumrate_linear} can be obtained by solving the following geometric program:
\begin{align}\label{eq:gp_temp}
    \textup{minimize} \;\;&\rho\nonumber\\
    \textup{subject to}\;\;&\textup{diag}(e^{\Tilde{\bm{\gamma}}})\B\z\leq \rho\z,\nonumber\\
    \textup{variables:\;\;\;}&\rho,\z,
\end{align}
which is convex through a logarithmic transformation. More concretely, let $\Tilde{\z} = \log \z$ and $\Tilde{\rho} = \log \rho$, and then the Lagrangian associate with the problem \eqref{eq:gp_temp} is 
\begin{align}
    \mathcal{L}_1(\Tilde{\rho},\Tilde{\z}, \bm{\nu}) = \Tilde{\rho} - \sum_l v_l \log(\sum_j b_{lj}e^{\Tilde{\gamma}_{l}+\Tilde{z}_j-\Tilde{z}_l-\Tilde{\rho}}),
\end{align}
where $\bm{\nu}=[\nu_1,\cdots,\nu_L]^T$ is the dual variable. 
Substituting \eqref{eq:gp_temp} into \eqref{eq:sumrate_linear},  then the Lagrange dual function  associate with \eqref{eq:sumrate_linear} is defined as
\begin{align}
   &g(\lambda,\bm{\nu})=\inf_{\Tilde{\rho},\Tilde{\z},\Tilde{\bm{\gamma}}} \mathcal{L}_2(\Tilde{\bm{\gamma}},\lambda,\Tilde{\rho},\Tilde{\z}, \bm{\nu}) \nonumber\\
   =& \inf_{\Tilde{\bm{\gamma}}}\Bigg(-\sum_lw_ly_{l2}^{(t)}\Tilde{\gamma}_{l}\nonumber\\
    &+ \lambda\inf_{\Tilde{\rho},\Tilde{\z}}\left(\Tilde{\rho} + \sum_l v_l \log\sum_j b_{lj}e^{\Tilde{\gamma}_{l}+\Tilde{z}_j-\Tilde{z}_l-\Tilde{\rho}}\right)\Bigg).
\end{align}
Taking the derivative of $\mathcal{L}_2(\Tilde{\bm{\gamma}},\lambda,\Tilde{\rho},\Tilde{\z}, \bm{\nu})$ with respect to $\Tilde{\bm{\gamma}},\Tilde{\rho},\Tilde{\z}$, we have
\begin{align}
    \frac{\partial \mathcal{L}_2(\Tilde{\bm{\gamma}},\lambda,\Tilde{\rho},\Tilde{\z}, \bm{\nu})}{\partial \Tilde{\gamma}_l}=&-w_ly_{l2}^{(t)}+\lambda v_l,\\
    \frac{\partial \mathcal{L}_2(\Tilde{\bm{\gamma}},\lambda,\Tilde{\rho},\Tilde{\z}, \bm{\nu})}{\partial \Tilde{\rho}}=&\lambda -\lambda \sum_i v_i,\\
    \frac{\partial \mathcal{L}_2(\Tilde{\bm{\gamma}},\lambda,\Tilde{\rho},\Tilde{\z}, \bm{\nu})}{\partial \Tilde{z}_l}=& \sum_i v_i \frac{b_{il} e^{\Tilde{z}_l}}{\sum_{j}b_{ij}e^{\Tilde{z}_j}} -v_l.
\end{align}
By setting the above derivatives as zeros, we have $\lambda  = \sum_l w_ly_{l2}^{(t)}$, $v_l =w_ly_{l2}^{(t)}/\sum_i w_iy_{i2}^{(t)}$ and
\begin{align}\label{eq:iter_temp}
{z}_l =  \frac{w_ly_{l2}^{(t)}}{\sum_i w_iy_{i2}^{(t)} \frac{b_{il}}{\sum_j b_{ij}{z}_j}}.
\end{align}
Since the right-hand side of \eqref{eq:iter_temp} is a concave function with respect of ${\mathbf{z}}$, according to Krause's theorem \cite{krause2001concave}, $\lim_{k\rightarrow \infty}\mathbf{z}^{(t)}$ in \eqref{eq:iteration} converges geometrically fast to its fixed point $\mathbf{z}^*$, and $t\mathbf{z}^*$ for any positive $t$ is an optimal solution to \eqref{eq:gp_temp} for $\mathbf{z}$, that is, the right  Perron-Frobenius eigenvector of $\textup{diag}(e^\r)\B$.
Then, once $\mathbf{z}^*$ is known, the optimal $\rho^*$ of \eqref{eq:gp_temp} can be regarded as a function of $\Tilde{\bm{\gamma}}$, that is,
\begin{align}\label{eq:rho_r}
    \rho^*(\Tilde{\bm{\gamma}})=\max\left\{   \frac{\left(\diag(e^{\Tilde{\bm{\gamma}}})\B\mathbf{z^*}\right)_l}{z_l^*}, l=1,\cdots,L\right\}.
\end{align}
On substituting \eqref{eq:rho_r} into \eqref{eq:sumrate_linear}, we can reformulate \eqref{eq:sumrate_linear} into the following equivalent problem:
\begin{align}\label{eq:sumrate_temp_v2}
\textup{maximize} \;\;\; &\sum_{l=1}^{L}w_ly_{l2}^{(t)}\Tilde{\gamma}_l \nonumber\\
\textup{subject to}  \;\;\;&\left(\diag(e^{\Tilde{\bm{\gamma}}})\B\mathbf{z^*}\right)_l\leq z^*_l, l=1,\cdots,L,\nonumber\\
\textup{variables:} \;\;\;&\Tilde{\gamma}_l, \;\forall\;l.
\end{align}
which is a convex optimization problem. Introducing the Lagrange multipliers $\bm{u}$ of \eqref{eq:sumrate_temp_v2}, we obtain the following KKT conditions:
\begin{align}
    u_l\geq 0, \quad u_l(e^{\Tilde{\gamma}_l}(\B\mathbf{z^*})_l- z^*_l) = 0, \quad l=1,\cdots,L,\nonumber
\end{align}
\begin{align}
    -w_ly_{l2}^{(t)}+u_l e^{\Tilde{\gamma}_l}(\B\mathbf{z^*})_l = 0,  \quad l=1,\cdots,L.\nonumber
\end{align}
Starting by substituting the complementary slackness conditions into the last equations, we can directly solve these KKT conditions, and find that $u_l^* = w_ly_{l2}^{(t)}/z_l^*$ and $\Tilde{\gamma}_l^* = \log(z^*_l/(\B\mathbf{z^*})_l)$.

\subsection{Proof of Lemma \ref{lem:lem2}}

It is easy to see that the maximizer for $y_{ln}$ in 
\begin{align}\label{eq:sumrate_temp_v3}
\textup{maximize} \;\;\; &\sum_{l=1}^{L}w_ly_{l2}^{(t)}\Tilde{\gamma}_l \nonumber\\
\textup{subject to}  \;\;\;&\left(\diag(e^{\Tilde{\bm{\gamma}}})\B\mathbf{z^*}\right)_l\leq z^*_l, l=1,\cdots,L,\nonumber\\
\textup{variables:} \;\;\;&\Tilde{\gamma}_l, \;\forall\;l.
\end{align}
is strictly convex; so the solution of \eqref{eq:sumrate_temp_v3} at each step is unique. Further, since the right-hand side of \eqref{eq:iter_temp} is also strictly concave, the iteration of \eqref{eq:iteration} has only one convergence point. Thus, the maximizing $\Tilde{\gamma}_l$ in \eqref{eq:sumrate_linear} is also unique. Then from \cite{tseng2001convergence} we have that the BCD method in Algorithm 1 will converge.

\subsection{Convergence Result}


To achieve convergence, the deterministic policy gradient method for continuous control, as stated in \cite{xiong2022deterministic}, must satisfy the following assumptions:

\begin{assumption}\cite{xiong2022deterministic}
     1) For any $\bm{\theta}_1, \bm{\theta}_2$, there exist positive constants $L_{\bm{\mu}}, L_\psi$ and $\lambda_{\Psi}$, such that $\left\|\bm{\mu}_{\bm{\theta}_1}(\bm{s})-\bm{\mu}_{\bm{\theta}_2}(\bm{s})\right\| \leq L_{\bm{\mu}}\left\|\bm{\theta}_1-\bm{\theta}_2\right\|, \forall  \bm{s} \in \bm{\mathcal{S}}$, $\|\nabla_{\bm{\theta}} \bm{\mu}_{\bm{\theta}_1}(\bm{s})- $ $\nabla_{\bm{\theta}} \bm{\mu}_{\bm{\theta}_2}(\bm{s})\| \leq L_\psi\left\|\bm{\theta}_1-\bm{\theta}_2\right\|, \forall  \bm{s} \in \bm{\mathcal{S}}$ and the matrix $\Psi_{\bm{\theta}}=\mathbb{E}_{\nu_{\bm{\theta}}}\left[\nabla_{\bm{\theta}} \bm{\mu}_{\bm{\theta}}(\bm{s}) \nabla_{\bm{\theta}} \bm{\mu}_{\bm{\theta}}(\bm{s})^T\right]$ is non-singular with the minimal eigenvalue uniformly lowerbounded as $\sigma_{\min }\left(\Psi_{\bm{\theta}}\right) \geq \lambda_{\Psi}$;
2) For any action $\bm{a}_1, \bm{a}_2$ in the action space $\bm{\mathcal{A}}$, there exist positive constants $L_P, L_r$, such that the transition kernel satisfies $\left|P\left(\bm{s}^{\prime} \mid \bm{s}, \bm{a}_1\right)-P\left(\bm{s}^{\prime} \mid \bm{s}, \bm{a}_2\right)\right| \leq L_P\left\|\bm{a}_1-\bm{a}_2\right\|, \forall  \bm{s}, \bm{s}^{\prime} \in \bm{\mathcal{S}}$ and the reward function satisfies $\left|r\left(\bm{s}, \bm{a}_1\right)-r\left(\bm{s}, \bm{a}_2\right)\right| \leq$ $L_r\left\|\bm{a}_1-\bm{a}_2\right\|, \forall\; \bm{s}, \bm{s}^{\prime} \in \bm{\mathcal{S}}$.
3) For any $\bm{a}_1, \bm{a}_2$ in $\bm{\mathcal{A}}$, there exists a positive constant $L_q$, such that $\left\|\nabla_a q\left(\bm{s}, \bm{a}_1\right)-\nabla_a q\left(\bm{s}, \bm{a}_2\right)\right\| \leq$ $L_q\left\|\bm{a}_1-\bm{a}_2\right\|, \forall  \bm{\theta} \in \mathbb{R}^d, \bm{s} \in \bm{\mathcal{S}}$.
4) The feature function $\bm{\phi}: \bm{\mathcal{S}} \times \bm{\mathcal{A}} \rightarrow$ $\mathbb{R}^d$ is uniformly bounded, i.e., $\|\bm{\phi}(\cdot, \cdot)\| \leq C_{\bm{\phi}}$ for some positive constant $C_{\bm{\phi}}$. In addition, we define $\mathbf{A}=$ $\mathbb{E}\left[\bm{\phi}(x)\left(\gamma \bm{\phi}\left(x^{\prime}\right)-\bm{\phi}(x)\right)^T\right]$, and assume that $\mathbf{A}$ is non-singular. Assume that the absolute value of the eigenvalues of $\mathbf{A}$ are uniformly lower bounded, i.e., $|\sigma(\mathbf{A})| \geq \lambda_{\mathbf{A}}$ for some positive constant $\lambda_{\mathbf{A}}$.
\end{assumption}

If our proposed DRL-based BCD can be configured to satisfy the above assumptions, then according to Theorem 1 in \cite{xiong2022deterministic}, the DRL-based BCD can achieve convergence. Therefore, our focus is to demonstrate that the neural networks and learning parameters in Algorithm \ref{alg:alg2} of the DRL-based BCD can satisfy the aforementioned assumptions. It is worth noting that the first assumption can be easily met by appropriately parameterized policy classes, such as the linear approximator. Assumption 3) establishes that the Q-function is smooth over action, which is a standard assumption in studies related to deterministic policies. Assumption 4) is a standard requirement in TD learning studies involving linear function approximation. The main contribution of our algorithm lies in the incorporation of the BCD algorithm to update the state $\bm{s}$. As stated in equations \eqref{eq:state}-\eqref{eq:policy}, the stochasticity of state transition originates from $\bm{\epsilon}\sim U(0, c)$ in equation \eqref{eq:policy}. Since problem \eqref{eq:sumrate_pro_r} has only a finite number of local optima, the number of states $\bm{s}$ is also finite. Consequently, the transition kernel is Lipschitz continuous. Furthermore, according to equation \eqref{eq:rl}, $\gamma_l(\p)\leq G_{ll}\bar{P}/n_l$ for any feasible $\p$. Thus, $\left|r\left(\bm{s}, \bm{a}_1\right)-r\left(\bm{s}, \bm{a}_2\right)\right| \leq \max_l 2G_{ll}\bar{P}/n_l$, indicating that the reward function is also Lipschitz continuous. Thus, Assumption 2) is satisfied. Consequently, all the assumptions in Assumption 1 can be satisfied, ensuring the convergence of the DRL-based BCD.




\bibliographystyle{IEEEtran}
\bibliography{_Reference_RL}




\end{document}